\documentclass[11pt,letterpaper]{article}
\usepackage{floatrow}
\usepackage{typearea}
\usepackage{pgfplots}
\usepackage{tikz}
\usetikzlibrary{patterns}

\usepackage{subcaption}
\usepackage{float}

\usepackage{bbm}
\usepackage[ruled,linesnumbered,vlined]{algorithm2e}
\typearea{14}

\usepackage{comment,multicol}

\usepackage[ruled,linesnumbered,vlined]{algorithm2e}
\usepackage{multirow}
\usepackage{wrapfig}
\usepackage{color,colortbl}
\usepackage{float}
\usepackage{amsthm}
\usepackage{amsmath}
\usepackage{amssymb}
\usepackage{graphicx}
\usepackage{xspace}
\usepackage{nicefrac}
\usepackage[colorinlistoftodos,prependcaption,textsize=tiny]{todonotes}

\definecolor{forestgreen}{rgb}{0.13, 0.55, 0.13}

\SetCommentSty{mycommfont}

{\makeatletter
	\gdef\xxxmark{%
		\expandafter\ifx\csname @mpargs\endcsname\relax % in minipage?
		\expandafter\ifx\csname @captype\endcsname\relax % in figure/caption?
		\marginpar{xxx}% not in a caption or minipage, can use marginpar
		\else
		xxx % notice trailing space
		\fi
		\else
		xxx % notice trailing space-
		\fi}
	\gdef\xxx{\@ifnextchar[\xxx@lab\xxx@nolab}
	\long\gdef\xxx@lab[#1]#2{{\bf [\xxxmark #2 ---{\sc #1}]}}
	\long\gdef\xxx@nolab#1{{\bf [\xxxmark #1]}}
}

\usepackage{dsfont}
\usepackage{thmtools}
\usepackage{thm-restate}
\declaretheorem[name=Lemma]{lemma}

\definecolor{Darkblue}{rgb}{0,0,0.4}
\definecolor{Brown}{cmyk}{0,0.61,1.,0.60}
\definecolor{Purple}{cmyk}{0.45,0.86,0,0}

\usepackage[colorlinks,linkcolor=Darkblue,filecolor=blue,citecolor=blue,urlcolor=Darkblue,pagebackref]{hyperref}
\usepackage[nameinlink]{cleveref}

\newcommand{\initOneLiners}{%
	\setlength{\itemsep}{0pt}
	\setlength{\parsep }{0pt}
	\setlength{\topsep }{0pt}
}
\newenvironment{OneLiners}[1][\ensuremath{\bullet}]
{\begin{list}
		{#1}
		{\initOneLiners}}
	{\end{list}}

\newcommand{\initTwoLiners}{%
	\setlength{\itemsep}{1pt}
	\setlength{\parsep }{0pt}
	\setlength{\topsep }{0pt}
}

\newtheorem{theorem}{Theorem}
\newtheorem{corollary}{Corollary}
\newtheorem{remark}{Remark}
\newtheorem{claim}{Claim}

\newtheorem{conjecture}{Conjecture}
\newtheorem{fact}{Fact}

\numberwithin{equation}{section}

\newcommand{\namedref}[2]{\hyperref[#2]{#1~\ref*{#2}}}

\newcommand{\wt}{\widetilde}

\newcommand{\R}{\mathbb{R}}

\newcommand{\poly}{{\rm poly}}
\newcommand{\polylog}{{\rm polylog}}

\newcommand{\etal}{{et al. \xspace}}

\newcommand{\eps}{\epsilon}

\usepackage{pdfsync}

\title{Graph Spanners by Sketching in Dynamic Streams and the Simultaneous Communication Model}
\author{Arnold Filtser\thanks{Supported by the Simons Foundation.}\\Columbia University \and Michael Kapralov\thanks{Supported in part by ERC Starting Grant 759471.}\\EPFL \and Navid Nouri\thanks{Supported by ERC Starting Grant 759471}\\EPFL}
\begin{document}
	%	\lipsum[1]
	\maketitle
	\thispagestyle{empty}
	%\nonumber
	\begin{abstract}
Graph sketching is a powerful technique introduced by the seminal work of Ahn, Guha and McGregor'12 on connectivity in dynamic graph streams that has enjoyed considerable attention in the literature since then, and has led to near optimal dynamic streaming algorithms for many fundamental problems such as connectivity, cut and spectral sparsifiers and matchings. Interestingly, however, the sketching and dynamic streaming complexity of approximating the shortest path metric of a graph is still far from well-understood. Besides a direct $k$-pass implementation of classical spanner constructions (recently improved to  $\lfloor\frac k2\rfloor+1$-passes by Fernandez, Woodruff and Yasuda'20) the state of the art amounts to a $O(\log k)$-pass algorithm of Ahn, Guha and McGregor'12, and a $2$-pass algorithm of Kapralov and Woodruff'14. In particular, no single pass algorithm is known, and the optimal tradeoff between the number of passes, stretch and space complexity is open.

In this paper we introduce several new graph sketching techniques for approximating the shortest path metric of the input graph. We give the first {\em single pass} sketching algorithm for constructing graph spanners: we show how to obtain a $\widetilde{O}(n^{\frac{2}{3}})$-spanner using $\widetilde{O}(n)$ space, and in general a $\widetilde{O}(n^{\frac{2}{3}(1-\alpha)})$-spanner using $\widetilde{O}(n^{1+\alpha})$ space for every $\alpha\in [0, 1]$, a tradeoff that we think may be close optimal. We also give new spanner construction algorithms for any number of passes, simultaneously improving upon all prior work on this problem. Finally, we note that unlike the original sketching approach of Ahn, Guha and McGregor'12, none of the existing spanner constructions yield {\em simultaneous communication} protocols with low per player information. We give the first such protocols for the spanner problem that use a small number of rounds.

	\end{abstract}
	\newpage
	%\vfill
	%{\small \setcounter{tocdepth}{1} \tableofcontents}
	
	\tableofcontents
\thispagestyle{empty}
	\newpage
	\setcounter{page}{1}
	\pagenumbering{arabic}

%!TEX root = ./SpannersDynamicStreams01.tex
\section{Introduction}
Graph sketching, introduced by~\cite{AGM12LinearMeasure} in an influential work on graph connectivity in dynamic streams has been a de facto standard approach to constructing algorithms for dynamic streams, where the algorithm must use a small amount of space to process a stream that contains both edge insertions and deletions. The main idea of~\cite{AGM12LinearMeasure} is to represent the input graph by its edge incident matrix, and applying classical linear sketching primitives to the columns of this matrix. This approach seamlessly extends to dynamic streams, as by linearity of the sketch one can simply subtract the updates for deleted edges from the summary being maintained: a surprising additional benefit is the fact that such a sketching solution is trivially parallelizable: since the sketch acts on the columns of the edge incidence matrix, the neighborhood of every vertex in the input graph is compressed independently. In particular, this yields efficient protocols in the {\em simultaneous communication} model, where every vertex knows its list of neighbors, and must communicate a small number of bits about this neighborhood to a coordinator, who then announces the answer. Surprisingly, several fundamental problems such as connectivity~\cite{AGM12LinearMeasure}, cut~\cite{AGM12Spanners} and spectral sparsification~\cite{DBLP:conf/approx/AhnGM13,DBLP:conf/focs/KapralovLMMS14,DBLP:conf/soda/KapralovMMMNST20} admit sketch based simultaneous communication protocols with only polylogarithmic communication overhead per vertex, which essentially matches existentially optimal bounds.\footnote{There is some overhead to using linear sketches, but it is only polylogarithmic in the number of vertices in the graph -- see~\cite{NelsonY19}. }
The situation is entirely different for the problem of approximating the shortest path metric of the input graph: it is not known whether existentially best possible space vs approximation quality tradeoffs can be achieved using a linear sketch. This motivates the main question that we study:

\begin{center}
\fbox{
\parbox{0.9\textwidth}{
\begin{center}
What are the optimal space/stretch/pass tradeoffs for approximating the shortest path metric using a linear sketch?
\end{center}
}
}
\end{center}

\paragraph{Sketching and dynamic streams.} Sketching is the most popular tool for designing algorithms for the dynamic streaming model. Sketching solutions have been recently constructed for many graph problems,
including spanning forest computation \cite{ahn2012analyzing}, cut and spectral sparsifiers \cite{DBLP:conf/approx/AhnGM13,DBLP:conf/focs/KapralovLMMS14,DBLP:conf/soda/KapralovMMMNST20},
spanner construction \cite{AGM12Spanners,KW14}, matching and matching size approximation \cite{AssadiKLY16,DBLP:conf/soda/AssadiKL17}, sketching the Laplacian \cite{10.1145/2840728.2840753,JS18} and many other problems. Also, results showing universality of sketching for this application are known, at least under some restrictions on the stream. The result of~\cite{DBLP:conf/stoc/LiNW14} shows such an equivalence under the assumption that the stream length is at least doubly exponential in the size of the graph. The assumption on the stream length was significantly relaxed for {\em binary} sketches, i.e., sketches over $\mathbb{GF}_2$, by~\cite{DBLP:conf/coco/HosseiniLY19,DBLP:conf/coco/KannanMSY18}. Very recently, it has been shown~\cite{DBLP:conf/stoc/KallaugherP20} that lower bounds on stream length are crucial for such universality results: the authors of~\cite{DBLP:conf/stoc/KallaugherP20} exhibit a problem with a sketching complexity, which is polynomial in the input size, that can be solved in polylogarithmic space on a short dynamic stream.

\paragraph{Spanners in the sketching model.}  A subgraph $H=(V, E)$ of a graph $G=(V, E)$ is a $t$-spanner of $G$ if for every pair $u, v\in V$ one has 
$$
d_G(u, v)\leq d_H(u, v)\leq t\cdot d_G(u, v)~,
$$
where $d_G$ stands for the shortest path metric of $G$ and $d_H$ for the shortest path metric of $H$. We assume in this paper that the input graph is unweighted, as one can reduce to this case using standard techniques at the expense of a small loss in space complexity.\footnote{Specifically,  one can  partition the input edges into geometric weight classes and run our sketch based algorithm on every class, paying a multiplicative loss in space bounded by the log of the ratio of the largest weight to the smallest weight. See also \cite{ES16,ADFSW19}.} For every integer $k\geq 1$, every graph $G=(V, E)$ with $n$ vertices admits a $(2k-1)$-spanner with $O(n^{1+1/k})$ edges, which is optimal assuming the Erd\H{o}s girth conjecture. 
The greedy spanner \cite{ADDJS93,FS20}, which is sequential by nature, obtain the optimal number of edges.
The  celebrated algorithm of Baswana and Sen~\cite{BS07} obtains $(2k-1)$-spanner with $\wt{O}(n^{1+1/k})$ edges. This algorithm consists of a sequence of $k$ clustering steps, and as observed by \cite{AGM12Spanners}, can be implemented in $k$ passes over the stream using the existentially optimal $\wt{O}(n^{1+1/k})$ space. A central question is therefore whether it is possible to achieve the existentially optimal tradeoff using fewer rounds of communication, and if not, what the optimal space vs stretch tradeoff is for a given number of round of communication. Prior to our work this problem was studied in~\cite{AGM12Spanners} and~\cite{KW14}. The former showed how to construct a $(k^{\log_{2}5}-1)$-spanner in $\log_2 k$ passes using space $\wt{O}(n^{1+1/k})$, and the latter showed how to construct a $(2^{k}-1)$-spanner in two passes and $\wt{O}(n^{1+1/k})$ space. In a single pass, the previously best known algorithm which uses $n^{1+o(1)}$ space is simply to construct a spanning tree, guaranteeing distortion $n-1$. 
Thus, our first question is:

\begin{center}
\fbox{
\parbox{0.9\textwidth}{
\begin{center}
In a single pass in the dynamic semi streaming model using $\wt{O}(n)$ space, is it possible to construct an $o(n)$ spanner?
\end{center}
}
}
\end{center}
We prove the following theorem in \Cref{sec:stream}, as a corollary we obtain a positive answer to the question above (as spectral sparsifier can be computed in a single dynamic stream pass \cite{DBLP:conf/focs/KapralovLMMS14}). 
\begin{restatable}{theorem}{thmSpectralSparsAreSpanners}\label{thm:SpectralSparsAreSpanners}
	Let $G=(V,E)$ be an undirected, unweighted graph. For a parameter $\eps\in(0,\frac{1}{18}]$, suppose that $H$ is a $(1\pm\eps)$-spectral sparsifier of $G$. Then $\widehat{H}$ is an $\tilde{O}(n^{\frac23})$-spanner of $G$, where $\widehat{H}$ is unweighted version of $H$. 
\end{restatable}

\begin{restatable}{corollary}{corSpectralSparsAreSpanners}\label{cor:SpannerBySparsifiers}
	There exists an algorithm that for any $n$-vertex unweighted graph $G$, the edges of which arrive in a dynamic stream, using $\tilde{O}(n)$ space, constructs a spanner with $O(n)$ edges and stretch $\wt{O}( n^{\frac23})$ with high probability.
\end{restatable}

Additionally, for the same setting, using similar techniques, we prove stretch $\wt{O}(\sqrt{m})$ (see \Cref{thm:SpectralSparsAreSpannersEdges}).

One might think that the polynomial stretch is suboptimal, but we conjecture that this is close to best possible, and provide a candidate hard instance for a lower bound in \Cref{sec:lb-conj}. Specifically,
\begin{conjecture}\label{conj:1}
Any linear sketch from which one can recover an $n^{2/3-\Omega(1)}$-spanner  with probability at least $0.9$ requires $n^{1+\Omega(1)}$ space.
\end{conjecture}

More generally, we give the following trade off between stretch and space in a single pass:
\begin{restatable}{corollary}{thmOnePassTradeoff}\label{thm:OnePassTradeoff}
	Consider an $n$-vertex unweighted graph $G$, the edges of which arrive in a dynamic stream. For every parameter $\alpha\in(0,1)$, there is an algorithm using $\wt{O}(n^{1+\alpha})$ space, constructs a spanner with stretch $\wt{O}( n^{\frac23(1-\alpha)})$ with high probability.
\end{restatable}
Similarly, for the same setting, we prove stretch $\wt{O}(\sqrt{m}\cdot n^{-\alpha})$ (see \Cref{thm:TradeOffOnePassEdges}).

Next, we consider the case when we are allowed to take more than one pass over the stream and ask the following question:
\begin{center}
\fbox{
\parbox{0.9\textwidth}{
\begin{center}
	For positive integers $k$ and $s$, what is the minimal $f_n(k,s)$ such that an $f_n(k, s)$-stretch spanner can be constructed using $s$ passes over a dynamic stream of updates to the input $n$-vertex graph using $\tilde{O}(n^{1+\frac1k})$ space?
\end{center}
}
}
\end{center}

We present two results, which together improve upon all prior work on the problem. At a high level both results are based on the idea of repeatedly contracting low diameter subgraphs and running a recursive spanner construction on the resulting supergraph. The main idea of the analysis is to carefully balance two effects: {\bf (a)} loss in stretch due to the contraction process and {\bf (b)} the reduction in the number of nodes in the supergraph. Since the number of nodes in the supergraphs obtained through the contraction process is reduced, we can afford to construct better spanners on them and still fit within the original stretch budget. A careful balancing of these two phenomena gives our results. Our first result uses a construction based on  a clustering primitive implicit in~\cite{KW14} (see \Cref{lem:KW_clustering}) and gives the best known tradeoff in at most $\log k$ passes:

\begin{restatable}{theorem}{thmKWInPasses}\label{thm:KWInPasses}
	For every real $k\in [1,\log n]$, and integer $g\in [1,\log k]$, there is a $g+1$ pass dynamic stream algorithm that given an unweighted, undirected $n$-vertex graph $G=(V,E)$, uses $\tilde{O}(n^{1+\frac1k})$ space, and computes w.h.p. a spanner $H$ with   $\tilde{O}(n^{1+\frac1k})$ edges and stretch $2\cdot(2^{\left\lceil (\frac{k+1}{2})^{\nicefrac{1}{g}}\right\rceil }-1)^{g}-1< 2^{g\cdot k^{\nicefrac{1}{g}}}\cdot2^{g+1}$.	
\end{restatable}

Using the same algorithm, while replacing the aforementioned clustering primitive from \Cref{lem:KW_clustering} with a clustering primitive from~\cite{BS07} (see \Cref{lem:BS_clusters}), we obtain the following result, which provides the best known tradeoff for more than $\log_k$ passes: 
\begin{restatable}{theorem}{thmBSInPasses}\label{thm:BSInPasses}
	For every real $k\in [1,\log n]$, and integer $g\in [1,\log k]$, there is a $g\cdot\left(\left\lceil (\frac{k+1}{2})^{\nicefrac{1}{g}}\right\rceil -1\right)+1<g\cdot k^{\nicefrac{1}{g}}+1$ pass dynamic stream algorithm that given an unweighted, undirected $n$-vertex graph $G=(V,E)$, uses $\tilde{O}(n^{1+\frac1k})$ space, and computes w.h.p. a spanner $H$ with  $\tilde{O}(n^{1+\frac1k})$ edges and stretch $2\cdot\left(2\cdot\left\lceil (\frac{k+1}{2})^{\nicefrac{1}{g}}\right\rceil -1\right)^{g}-1\approx2^{g}\cdot\left(k+1\right)$.	
\end{restatable}
The proofs of \Cref{thm:KWInPasses} and \Cref{thm:BSInPasses} are presented in \Cref{sec:spanner-sketch}. We present our improvements over prior work in Table~\ref{tab:results} below, where the results of \Cref{thm:KWInPasses} and \Cref{thm:BSInPasses} are present via several corollaries, presented in \Cref{sec:corollaries}.

\begin{figure}[h]
	\setlength{\columnseprule}{0pt}
	\begin{multicols}{2}
		\renewcommand{\arraystretch}{1.3}	
	\scalebox{0.8}{
		\begin{tabular}{|l|l|l|l|}
		\hline
		\textbf{\#Passes}        & \textbf{Space}             & \textbf{Stretch}                            & \textbf{Reference}                 \\ \hline		
		$1$                      & $\tilde{O}(n)$             & $\tilde{O}(n^{\nicefrac23})$                    & \Cref{cor:SpannerBySparsifiers} \\ \hline
		$1$                      & $\tilde{O}(n^{1+\frac1k})$             & $\tilde{O}(n^{\nicefrac23(1-\frac1k)})$                    & \Cref{thm:OnePassTradeoff} \\ \hline\hline
		
		$2$                      & $\tilde{O}(n^{1+\nicefrac1k})$ & $2^k-1$                                     & \cite{KW14}                         \\ \hline
		$2$                      & $\tilde{O}(n^{1+\nicefrac1k})$ & $2^{\frac{k+3}{2}}-3$                                     & \Cref{cor:BetterKW}                        \\ \hline
		$g+1$                    & $\tilde{O}(n^{1+\nicefrac1k})$ & $2^{g\cdot k^{\nicefrac{1}{g}}}\cdot2^{g+1}$ & \Cref{thm:KWInPasses}           \\ \hline\hline
		
		$\log k$                 & $\tilde{O}(n^{1+\nicefrac1k})$ & $k^{\log5}-1$                             & \cite{AGM12Spanners}                \\ \hline
		$\log (k+1)$                 & $\tilde{O}(n^{1+\nicefrac1k})$ & $2\cdot k^{\log 3}-1$                             & \Cref{cor:logKpasses}                  \\ \hline\hline
		$g\cdot k^{\nicefrac{1}{g}}+1$ & $\tilde{O}(n^{1+\nicefrac1k})$ & $\approx2^{g}\cdot\left(k+1\right)$       & \Cref{thm:BSInPasses}             \\ \hline
		
		$k$                      & $\tilde{O}(n^{1+\nicefrac1k})$ & $2k-1$                                      & \cite{BS07}          \\ \hline
		$\lfloor\frac{k}{2}\rfloor+1$  & $\tilde{O}(n^{1+\nicefrac1k})$ & $2k-1$&\Cref{cor:halfBS}, \cite{FWY20}\\\hline

\end{tabular}} \par 
		\includegraphics[scale=.4]{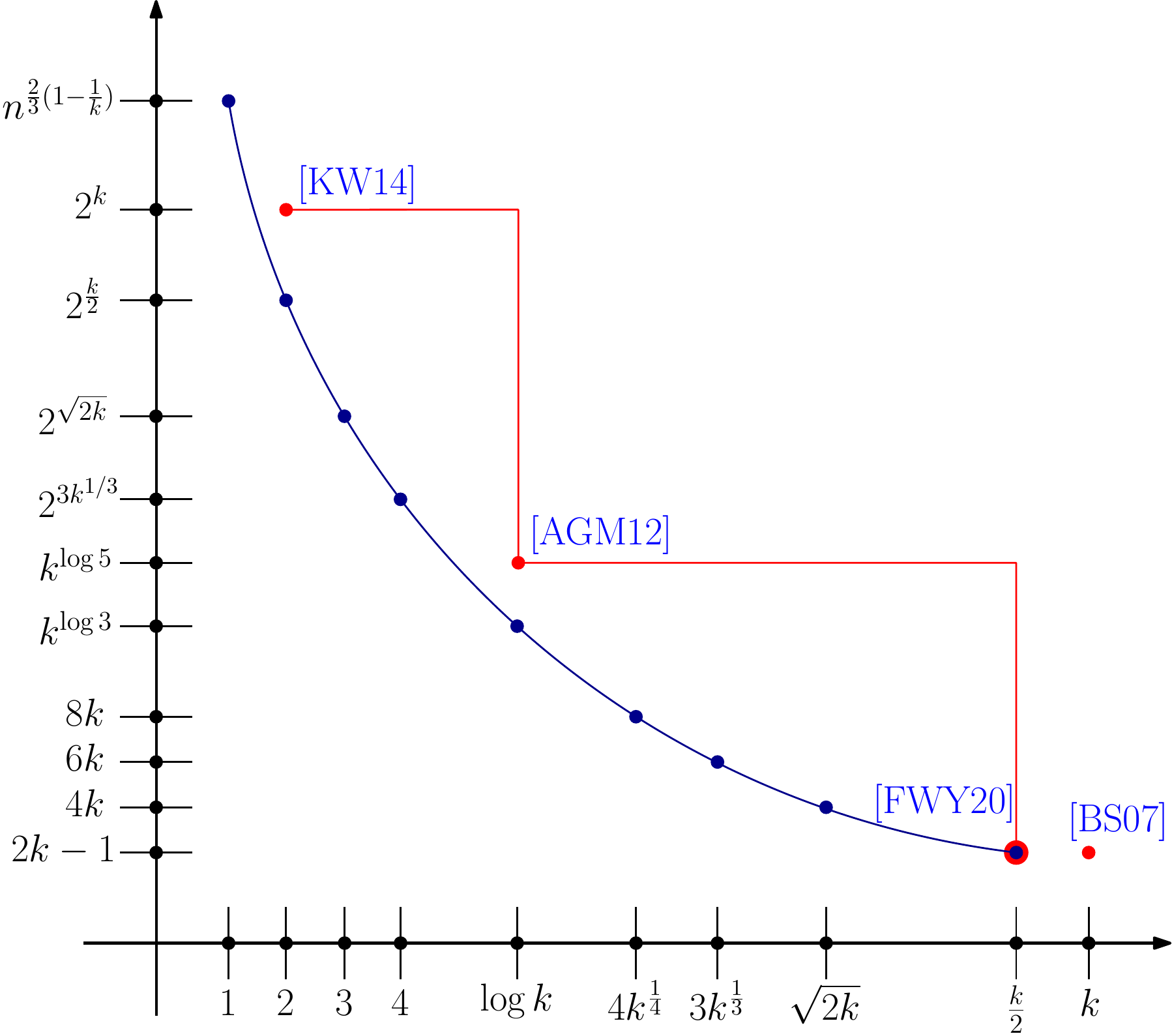}\par 
	\end{multicols}
	\vspace{-10pt}
	\caption{Trade-offs of various algorithms between stretch to number of passes. On the left summery of previous and current results. 
	In \Cref{thm:KWInPasses} and \Cref{thm:BSInPasses} $g$ can be chosen to be any integer in $[1,\log k]$.
	On the right in blue is a plot of the results in this paper, while the previous results depicted in red. All the results are for algorithms using $\tilde{O}(n^{1+\frac1k})$ space. The $Y$ axis represents the stretch, while the $X$ axis is the number of passes.
	Some second order terms are neglected.
}
	\label{tab:results}

\end{figure}

\paragraph{Simultaneous communication model.}  We also consider the related simultaneous communication model\footnote{This model has also been referred to as {\bf distributed sketching} in the literature (see e.g., \cite{NelsonY19}).}, which we now define. In the {\em simultaneous communication model} every vertex of the input graph $G=(V, E), |V|=n,$ knows its list of neighbors (so that every $e=(u, v)\in E$ is known to  both $u$ and $v$), and all vertices have a source of shared randomness. Communication proceeds in rounds, where in every round the players simultaneously post short messages on a common board for everyone to see (note that equivalently, one could consider a coordinator who receives all the messages in a given round, and then posts a message of unbounded length on the board). Note that a given player's message in any given round may only depend on their input and other players' messages in previous rounds. The content of the board at the end of the communication protocol must reveal the answer with high constant probability. The {\em cost} of a protocol in the simultaneous communication model is the length of the longest message communicated by any player.

Sketching algorithms for dynamic connectivity and cut/spectral approximations based on the idea of applying a sketch to the edge incidence matrix of the input graph~\cite{AGM12LinearMeasure,DBLP:conf/focs/KapralovLMMS14,DBLP:conf/soda/KapralovMMMNST20} immediately yield efficient single pass simultaneous communication protocols with only polylogarithmic message length. We note, however, that existing sketch based algorithms for spanner construction (except for our result in \Cref{cor:SpannerBySparsifiers} and the trivial $k$-pass implementation of the algorithm of Baswana and Sen~\cite{BS07}) {\bf do not} yield low communication protocols. This is because they achieve reductions in the number of rounds by performing some form of leader election and amortizing communication over all vertices. See \Cref{rem:ImpossibleKWinSCM} in \Cref{sec:KW} for more details.  To illustrate the difference between dynamic streaming and simultaneous communication model, consider the following artificial problem. Suppose that we are given a graph with all but $\sqrt{n}$ isolated vertices, and the task is to recover the induced subgraph. Then using sparse recovery we can recover all the edges between the vertices using a single pass over a dynamic stream of updates. However, it is clear from information theoretic considerations that a typical vertex will need to communicate $\Omega(\sqrt{n})$ bits of information to solve this problem. 
	
The main result of \Cref{sec:filtering} is the following theorem. 
	
\begin{restatable}{theorem}{thmfiltering}\label{thm:filtering}
		For any integer $g\ge1$, there is an algorithm (see \Cref{alg:filttering}) that in $g$ rounds of communication outputs a spanner with stretch $\min\left\{\tilde{O}(n^{\frac{g+1}{2g+1}}),~(12+o(1))\cdot n^{2/g}\cdot \log n\right\}$. 
\end{restatable}
Note that when $g=1$, the above theorem gives a $\wt{O}(n^{2/3})$ approximation using polylogarithmic communication per vertex. We think that the $n^{2/3}$ approximation is likely best possible in polylogarithmic communication per vertex, and the same candidate hard instance from \Cref{sec:lb-conj} that we propose for \Cref{conj:1} can probably be used to obtain a matching lower bound. Analyzing the instance appears challenging due to the fact that every edge is shared by the two players -- exactly the feature of our model that underlies our algorithmic results (this sharing is crucial for both connectivity and spectral approximation via sketches). This model bears some resemblance to the number-on-the-forehead  (NOF) model in communication complexity (see, for example,~\cite{KallaugherMPV19}, where a connection of this form was made formal, resulting in conditional hardness results for subgraph counting in data streams).

The proof of this theorem is presented in \Cref{sec:filtering}.

Additionally, in \Cref{subsec:SimModelTradeoff} we first provide a trade off between size of the communication per player and stretch in one round of communication.

\begin{restatable}{theorem}{thmSCMtradeoffoneround}\label{thm:SCmodel-1round}
	There is an algorithm that in $1$ round of communication, where each player communicates $\tilde{O}(n^\alpha)$ bits, outputs a spanner with stretch $$\min\left\{\wt{O}(n^{(1-\alpha)\frac{2}{3}}),~\wt{O}\left(\sqrt{m} \cdot n^{-\alpha}\right)\right\}~.$$ 
\end{restatable}

Then, we also prove a similar trade off when more than one round of communication is allowed. 

\begin{restatable}{theorem}{thmSCMtradeoff}\label{thm:SCmodel} 
		For any integer $g\ge 1$, there is an algorithm that in  $g$ rounds of communication, where each player communicates $\tilde{O}(n^\alpha)$ bits, outputs a spanner with stretch $$\min\left\{(12+o(1))\cdot n^{(1-\alpha)\cdot \frac2g}\cdot \log n~,~\wt{O}\left(n^{\frac{(g+1)(1-\alpha)}{2g+1}}\right)\right\}~.$$ 
\end{restatable}

\paragraph{Parallel work.} Very recently, in a paper about the message-passing model, Fernandez \etal \cite{FWY20} implemented Baswana-Sen \cite{BS07} algorithm in $\lfloor\frac k2\rfloor+1$ passes in the semi-streaming model. This is the same as our \Cref{cor:halfBS} (which follows from \Cref{thm:BSInPasses} by setting $g=1$). While writing this paper, the authors were not aware of \cite{FWY20} result.

The idea of recursively constructing a spanner by contracting clusters, which is the main idea leading to our \Cref{thm:KWInPasses,thm:BSInPasses}, was found and used concurrently and independently from us by Biswas \etal \cite{BDGMN20} in the context of the massive parallel computation (MPC) model. 

\paragraph{Related work.} Streaming algorithms are well-studied with too many results to list and we refer the reader to \cite{McGregor14,McGregor17} for a survey of streaming algorithms. The idea of linear graph sketching was introduced in a seminal paper of Ahn, Guha, and McGregror \cite{ahn2012analyzing}. An extension of the sketching approach to hypergraphs were presented in~\cite{GuhaMT15}.  The simultaneous communication model has also been used for lower bounding the performance of sketching algorithms -- see, e.g. ~\cite{AssadiKLY16,KallaugherKP18}.

Spanners are a fundamental combinatorial object. They have been extensively studied and have found numerous algorithmic applications. We refer to the survey \cite{ABD+19} for an overview.
The most relevant related work is on insertion only streams \cite{Elkin11,Baswana08} where the focus is on minimizing the processing time of the stream, and dynamic algorithms, where the goal is to efficiently maintain a spanner while edges are continuously inserted and deleted \cite{Elkin11,BKS12,BFH19}. 

\section{Preliminaries}\label{sec:prelims}
All the logarithms in the paper are in base $2$. 
We use $\tilde{O}$ notation to suppress constants and poly-logarithmic factors in $n$, that is $\widetilde{O}(f)=f\cdot\polylog(n)$.

We consider undirected, graphs $G=(V,E)$, with a weight function
$w: E \to \R_{\ge 0}$. If we say that a graph is unweighted, we mean that all the edges have unit weight.
$\widehat{G}=(V,E,\mathbbm{1}_E)$ denotes the unweighted version of $G$, i.e. the graph $G$ where all edge weights are changed to $1$.
Sometimes we abuse notation and write $G$ instead of $E$.
Given two subsets $X,Y\subseteq V$, $E_{G}(X,Y)$ is the set of edges from $X$ to $Y$, $w_G(X,Y)$ denotes the total weight of edges in $E_{G}(X,Y)$ (number if $G$ is unweighted).
We sometimes abuse notation and write instead $E_{G}(X\times Y)$ and $w_{G}(X\times Y)$ (respectively).
For a subset of vertices $A\subseteq V$, let $G[A]$ denote the induced graph on $A$.

Let $d_{G}$ denote the shortest path metric in $G$.
A subgraph $H$ of $G$ is a $t$-\emph{spanner} of $G$ if for every $u,v\in V$, $d_H(u,v)\le t\cdot d_G(u,v)$ (note that as $H$ is a subgraph of $G$, necessarily $d_G(u,v)\le d_H(u,v)$).
Following the triangle inequality, in order to prove that $H$ is a $t$-spanner of $G$ it is enough to show that for every edge $(u,v)\in E$, $d_H(u,v)\le t\cdot d_G(u,v)$.

For an unweighted graph $G=(V,E)$, such that $|V|=n$ and $|E|=m$, let $B_G\in \R^{m\times n}$ denote the vertex edge incidence matrix. The Laplacian matrix of $G$ is defined as $L_G:=B_G^\top B_G$. Similarly, for a weighted graph $H=(V,E,w)$, we let $W\in \R^{m\times m}$ be the diagonal matrix of the edge weights. The Laplacian of the graph $H$ is defined as $L_H:=B_H^\top W B_H$. 
$H\preceq G$ denotes that for every $\vec{x}\in\R^n$, $\vec{x}^t L_H\vec{x}\le \vec{x}^t L_G\vec{x}$.
We say that a graph $H$ is $(1\pm \epsilon)$-\emph{spectral sparsifier} of a graph $G$, if $$(1-\epsilon)H\preceq G\preceq (1+\epsilon) H~.$$ 
\begin{fact}\label{fact:cutsp}
	Suppose that a graph $H$ is a $(1\pm \epsilon)$-spectral sparsifier of a graph $G$, then $H$ is a $(1\pm \epsilon)$-cut sparsifier of $G$, i.e., for every set of vertices $S\subset V$, we have
	\begin{align*}
	(1-\epsilon)\cdot w_H(S,V\setminus S)~\le~ w_G(S,V\setminus S) ~\le~ (1+\epsilon)\cdot w_G(S,V\setminus S)~.
	\end{align*}
\end{fact}

For any Laplacian matrix $L_G$, we denote its Moore-Penrose
pseudoinverse by $L_G^{+}$. For any pair of vertices $u,v\in V$, we denote their indicator  vector by $b_{uv} = \chi_u - \chi_v$, where $\chi_u \in \R^{n}$ is the indicator vector of $u$, i.e., the entry corresponding to $u$ is $+1$ and all other entries are zero. Also, for any edge $e=(u,v)$, we define its indicator vector as $b_e:=b_{uv}$. We also define effective resistance of a pair of vertices $u,v\in V$ as
\begin{align*}
R_{uv}^G := b_{uv}^\top L^+_G b_{uv}. 
\end{align*}

\begin{fact}\label{fact:EffSparse}
	Given a $(1\pm \epsilon)$-spectral sparsifier $H$ of a $G$,  for every $u,v\in V$ it holds that
	\begin{align*}
	(1-\epsilon)R_{uv}^G~\le~ R_{uv}^{H} ~\le~ (1+\epsilon)R_{uv}^G.
	\end{align*}
\end{fact}
The following fact is a standard fact about effective resistances (see e.g., \cite{DBLP:conf/stoc/SpielmanS08})
\begin{fact}\label{fact:sumeff}
	In every $n$ vertex graph $G=(V,E,w)$ it holds that
	$\sum_{e\in E}w_eR_e^G ~\le~ n-1$.\footnote{If graph $G$ is connected, then the inequality is satisfied by equality.}

\end{fact}

\paragraph{Dynamic streams.} In dynamic streams, there is a fixed set $V$ of $n$ vertices, unweighted edges arrive in a streaming fashion, where they are both inserted and deleted.

\emph{$\ell_0$-samplers.}:
Given integer vector in $\R^n$ in a dynamic stream, using $s\cdot\polylog(n)$ space, we can sample $s$ different non-zero entries. In particular if the vector is $s$-sparse, we can reconstruct it.
Furthermore, given a stream of edges in an $n$-vertex graph $G$,
using $s\cdot \polylog(n)$ samplers per vertex, we can create a subgraph
$\tilde{G}$ of $H$ where each vertex has either at least $s$ edges, or has all its incident edges from $G$.
This samplers are linear, therefore if we sum up the samplers of $S$
vertices, we can sample an outgoing edge.

Consider a vector $\vec{v}\in\mathbb{R}^n$, given a subset $A\subseteq[n]$ of coordinates, we denote by $\vec{v}[A]$ the restriction of $\vec{v}$ to $A$.
\begin{lemma}\label{lem:SubsetSampling}
	Consider a vector $\vec{v}\in\mathbb{R}^n$ that arrives in a dynamic stream via coordinate updates. The coordinates $[n]$ are partitioned into subsets $A_1,A_2,\dots,A_r$ (the space required to represent this partition is negligible).
	Let $\mathcal{I}=\left\{i\mid\vec{v}[A_i]\ne\vec{0}\right\}$ be the indices of the coordinate sets  on which $\vec{v}$ is not zero.
	Given $A_1,A_2,\ldots, A_r$ and a parameter $s>0$, and a guarantee that 	
	$\left|\mathcal{I}\right|\le s$, using $s\cdot\polylog(n)$ space, one can design a sketching algorithm recovering a set $S\subseteq [n]$ such that
	\begin{itemize}
		\item For every $j\in S$, $\vec{v}_j\ne0$.
		\item For every $i\in\mathcal{I}$, $A_i\cap S\ne\emptyset$.
	\end{itemize}
\end{lemma}
The proof uses a technique commonly used in sketching literature, and is given in \Cref{sec:omitted-2} for completeness. 
\begin{lemma} \label{lem:RecoverSingleEdge}[Edge recovery]
	Consider an unweighted, undirected graph $G=(V,E)$ that is received in a dynamic stream. Given $A,B\subseteq V$ such that $A\cap B = \emptyset$, one can design a sketching algorithm that using $\polylog(n)$ space in a single pass over the stream, with probability $1/\poly(n)$, can either recover an edge between $A$ to $B$, or declare that there is no such edge.\\
	Further, provided that there are at most $m$ edges in $A\times B$, using $m\cdot \polylog(n)$ space, with probability $1/\poly(n)$ we can recover them all.
\end{lemma}
The proof is using the same techniques as in the proof of \Cref{lem:SubsetSampling} and is deferred to \Cref{sec:omitted-2}.

\section{Technical Overview}
We consider an $n$ vertex unweighted graph $G=(V,E)$.
\paragraph{Spectral sparsifiers are spanners (\Cref{sec:stream}).}
\begin{wrapfigure}{r}{0.35\textwidth}
	\begin{center}
		\vspace{-20pt}
		\includegraphics[width=0.9\textwidth]{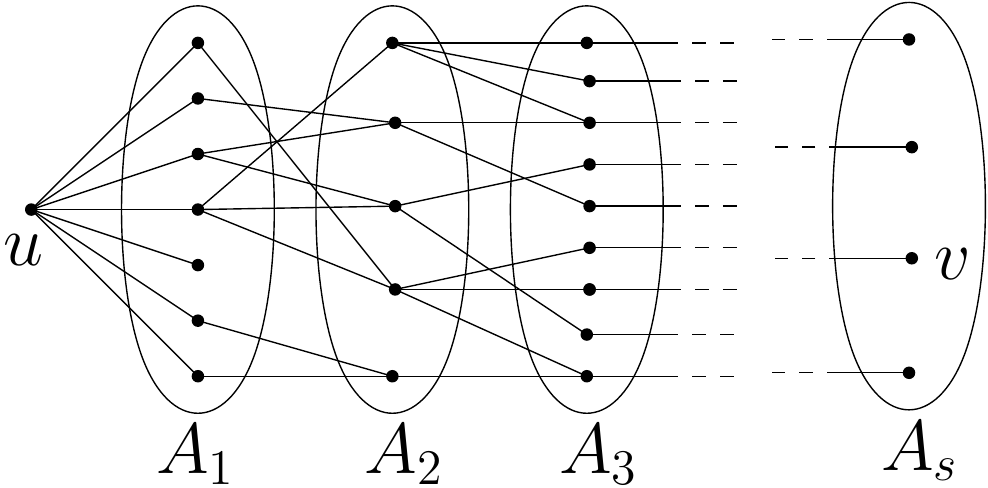}
		\vspace{-5pt}
	\end{center}
	\vspace{-10pt}
\end{wrapfigure}
The technical part of the paper begins by proving the following fact: consider a spectral sparsifier $H$ of $G$.
Consider an edge $(u,v)\in E$. Denote the distance between its endpoints in $\widehat{H}$ by $d_{\widehat{H}}(u,v)=s$. Divide the vertices $V$ into the BFS layers w.r.t. $u$ in $\widehat{H}$. That is, $A_i$ is the set of all vertices at distance $i$ from $u$ in $\widehat{H}$. In particular $v\in A_s$. See illustration on the right. 
Let $W_i^G=w_G(A_i\times A_{i+1})$ be the total weight of the edges in $E_G(A_i, A_{i+1})$. Similarly $W_i^H=w_H(A_i\times A_{i+1})$. Let $H'$ be the graph created from $H$ by contracting all the vertices in each set $A_i$ into a single vertex. The rough intuition is the following:
\begin{equation}
1\overset{(a)}{\ge}R_{u,v}^{G}\overset{(b)}{\gtrsim}R_{u,v}^{H}\overset{(c)}{\ge}R_{u,v}^{H'}\overset{(d)}{=}\sum_{i=0}^{s-1}\frac{1}{W_{i}^{H}}\overset{(*)}{\approx}\sum_{i=0}^{s-1}\frac{1}{W_{i}^{G}}\overset{(e)}{\ge}\sum_{i=0}^{s-1}\frac{1}{|A_{i}||A_{i+1}|}\overset{(f)}{\ge}\Omega\left(\frac{s^{3}}{n^{2}}\right)~.\label{eq:intuition}
\end{equation}
Here (a) follows as the effective resistance between the endpoints of an edge is at most $1$. 
(b) as $H$ is a spectral sparsifier of $G$.
(c) as the effective resistance can only reduce by contracting vertices.
(d) as $H'$ is a path graph.
(e) as $G$ is unweighted and thus $W_i^G$ is bounded by the number of edges in $A_i\times A_{i+1}$.
And (f)  as $\sum_i|A_i|\le n$ and the function $\sum_{i=0}^{s-1}\frac{1}{|A_{i}||A_{i+1}|}$ is minimized when $|A_i|=\Omega(\frac{n}{s})$ for all $i$.
The tricky part is the rough equality (*). Note that if \Cref{eq:intuition} holds, it will follow that $s=O(n^{\frac23})$, implying the desired stretch.

While $H$ is a spectral sparsifier of $G$, $W_i^G$ does not represent the size of a cut in $G$. This is as there might be edges in $G$ crossing from $A_i$ to $\cup_{j> i+1}A_j$, or from $A_{i+1}$ to $\cup_{j<i}A_j$. Thus a priori there is no reason to expect that $W_i^H$ will approximate $W_i^G$.
Interestingly, we were able to show that $W_{i}^{G}=W_{i}^{H}\pm\epsilon\cdot(W_{i-1}^{H}+W_{i}^{H}+W_{i+1}^{H})$. That is, while we are not able to bound $|W_i^G-W_i^H|$ using the standard  factor $\epsilon\cdot W_{i}^{H}$, we can bound this error once we take into account also the former and later cuts in the BFS order!
We use this fact to show that for most of the indices $i$, $W_i^H\le |A_i||A_{i+1}|$. The desired bound follows. See proof of \Cref{thm:SpectralSparsAreSpanners} for more details. 

Next, using similar analysis we show that in case where the graph $G$ has $m$ edges, the stretch of $\widehat{H}$ is bounded by $O(\sqrt{m})$ (see \Cref{thm:SpectralSparsAreSpannersEdges}). 
Suppose that $d_{\widehat{H}}(u,v)=s$. Intuitively, following \Cref{eq:intuition}, as $\sum_iW_i^G\le m$, it follows that $1\ge R_{u,v}^{G}\gtrsim\sum_{i=0}^{s-1}\frac{1}{W_{i}^{G}}=\Omega\left(\frac{s^{2}}{m}\right)$ (as $\sum_{i=0}^{s-1}\frac{1}{W_{i}^{G}}$ us minimized when all $W_i^G$'s are equal), implying $s=O(\sqrt{m})$.
Both bounds ($O(n^{\frac23})$ and $O(\sqrt{m})$) are tight. 
Essentially, we construct the exact instance tightening all the inequalities in \Cref{eq:intuition}. That is a graph with $\widetilde{\Theta}(n^{\frac23})$ layers, each one containing $\widetilde{\Theta}(n^{\frac13})$ vertices, and all possible edges between layers (see \Cref{subsec:LowerBound}).

In \Cref{sec:streamingtradeoff}, we show that using $\widetilde{O}(n^{1+\alpha})$ space (instead of  $\widetilde{O}(n)$), the stretch can be reduced to $$\min\{\widetilde{O}(n^{\frac23(1-\alpha)}),\widetilde{O}(\sqrt{m}\cdot n^{-\alpha})\}.$$ The idea is the following: randomly partition the graph $G$ into $\widetilde{O}(n^{2\alpha})$ induced subgraphs $G_1,G_2,\dots$, such that each $G_i$ contains $O(n^{1-\alpha})$ vertices, and every pair of vertices $u,v$ belong to some $G_i$. Furthermore, the (expected) number of edges in each $G_i$ is $m\cdot n^{-2\alpha}$. Next, we construct a spectral sparsifier for each graph $G_i$ and take their union as our spanner. The stretch gurantee follows (see \Cref{thm:meta-space-stretch}, \Cref{thm:OnePassTradeoff} and \Cref{thm:TradeOffOnePassEdges}).

\paragraph{Simultaneous communication model (\Cref{sec:CongestedMultiplayerModel}).}
In a single pass, one can construct a spectral sparsifier and therefore obtain the exact same results as in the streaming model.
However, as opposed to streaming, no known approach can reduce the stretch in less than logarithmic number of rounds.
We propose a natural peeling algorithm (see \Cref{alg:filttering}). Denote $G_1=G$. Given a desired stretch parameter $t$, the algorithm computes a spectral sparsifier $H_1$, and removes all the satisfied edges $(u,v)\in E$ where $d_{\widehat{H}_1}(u,v)\le t$, to obtain a graph $G_2$. Generally, in the $i$'th round the algorithm computes a spectral sparsifier $H_i$ for the graph $G_i$, and removes all the satisfied edges to obtain $G_{i+1}$. This procedure continues until all the edges are satisfied (that is $G_{i+1}=\emptyset$).
The resulting spanner is $\widehat{H}=\cup_i\widehat{H}_i$ the union of (the unweighted version of) all the constructed sparsifiers.
Notably, for every parameter $t\ge 1$ the algorithm will eventually halt, and return a $t$-spanner. The arising question is, how many rounds are required to satisfy a specific parameter $t$?

We show that this procedure will halt after $g$ steps for $$t\ge\min\{\tilde{O}(n^{\frac{g+1}{2g+1}})~,~(12+o(1))\cdot n^{2/g}\cdot \log n\}$$
(see \Cref{thm:filtering}).
Interestingly, in $g=\log n$ rounds we can obtain stretch $O(\log n)$, which is asymptotically optimal.
That is, we present a completely new construction  for a $O(\log n)$-spanner with $\tilde{O}(n)$ edges.
Interestingly, there are constructions of spectral sparsifiers which are based on taking a union of poly-logarithmically many $O(\log n)$-stretch spanners (see \cite{KP12,KX16}). In a sense, here we obtain the opposite direction. That is, by taking a union of $\log n$ sparsifiers, one can construct an $O(\log n)$ stretch spanner. That is, sparsifiers and spanners are much more related from what one may initially expect.

To show that the algorithm halts in $g$ round for a specific $t$, we bound the number of edges in $G_i$, which eventually will lead us to conclusion that $G_{g+1}=\emptyset$:
\begin{itemize}
	\item  Set $t=\tilde{O}(n^{\frac{g+1}{2g+1}})$. Here the analysis is based on the effective resistance. 
	Using \Cref{eq:intuition}, one can see that after the first round, $G_2$ will contain only edges with effective resistance at least $\Omega(\frac{t^3}{n^2})$ (in $G$). As the sum of all effective resistances is bounded by $n-1$, we conclude $|G_2|\le \Omega(\frac{n^3}{t^3})$.
	In general, following the $O(\sqrt{m})$ upper bound on stretch, one can show that $G_{i+1}$ contain only edges with effective resistance  $\Omega(\frac{t^2}{|G_i|})$, implying $|G_{i+1}|\le \frac{n}{t^2}|G_i|$. $t$ is chosen so that $|G_g|\le t^2$, hence a spectral sparsifier will have stretch at most $\sqrt{|G_g|}=t$ for all the edge, implying $G_{g+1}=\emptyset$.
	\item Set $t=O( n^{2/g}\cdot \log n)$. Here the analysis is based on low diameter decomposition. In general, for a weighted graph $H$ and parameter $\phi=n^{-2/g}$, we construct a partition  $\mathcal{C}$ of the vertices, such that each cluster $C\in\mathcal{C}$ has hop-diameter $O(\frac{\log n}{\phi})=t$ (i.e. w.r.t. $\widehat{H}$), and the overall fraction of the weight of inter-cluster edges is bounded by $\phi$. Following our peeling algorithm, when this clustering is preformed w.r.t. $H_i$, $G_{i+1}$ will contain only inter-cluster edges from $G_i$. As $H_i$ is a spectral sparsifier of $G_i$, the size of all cuts are preserved. It follows that $|G_{i+1}|\lesssim\phi\cdot |G_{i}|$. In particular, in $\log_{\frac{1}{\phi}}|G|\le g$ rounds, no edges will remain.\\
	Interestingly, for this analysis to go through it is enough that each $H_i$ will be a cut sparsifier of $G_i$, rather than a spectral sparsifier. Oppositely, a single cut sparsifier $H$ of $G$ can have stretch $\tilde{\Omega}(n)$ (see \Cref{rem:CutSparsifier}).
\end{itemize}

Next, similarly to the streaming case, we show that if each player can communicate a message of size $\widetilde{O}(n^\alpha)$ in each round, then we can construct a spanner with stretch $\min\{\widetilde{O}(n^{\frac23(1-\alpha)}),\widetilde{O}(\sqrt{m}\cdot n^{-\alpha})\}$ in a single round, or stretch $\min\left\{(12+o(1))\cdot n^{(1-\alpha)\cdot \frac2g}\cdot \log n~,~\wt{O}\left(n^{\frac{(g+1)(1-\alpha)}{2g+1}}\right)\right\}$ in $g$ rounds (see \Cref{thm:SCmodel-1round} and \Cref{thm:SCmodel}).
The approach is the same as in the streaming case, and for the most part, the analysis follows the same lines.
However, the single round $\widetilde{O}(\sqrt{m}\cdot n^{-\alpha})$ bound is somewhat more involved. Specifically, in the streaming version we've made the assumption that $m\le n^{1+\alpha}$, as otherwise, using sparse recovery we can restore the entire graph. Unfortunately, sparse recovery is impossible here. Instead, we show that in a single communication round we can partition the vertex set $V$ into $V_1,V_2$, such that all the incident edges of $V_1$ are restored, while the minimum degree in $G[V_2]$ is at least $n^{\alpha}$. The rest of the analysis goes through.

\paragraph{Pass-stretch trade-off (\Cref{sec:spanner-sketch}).}
Fix the allowed space of the algorithm to be $\widetilde{O}(n^{1+\frac1k})$. Both \cite{BS07} and \cite{KW14} algorithms are based on clustering. Specifically, they have $k$ clustering phases, where in the $i$'th phase there are about $n^{1-\frac ik}$ clusters. Eventually, after $k-1$ phases the number of clusters is $n^{\frac 1k}$, and an edge from every vertex to every cluster could be added to the spanner. 
In \cite{BS07}, each clustering phase takes a single dynamic stream pass, while the diameter of each $i$-level cluster is bounded by $2i$. 
On the other hand, in \cite{KW14} all the clusters are constructed in a single dynamic stream pass, while the diameter of each $i$-level clusters is only bounded by $2^{i+1}-2$. 

Our basic approach is the following: execute either \cite{BS07} or \cite{KW14} clustering procedure for some $i$ steps. Then, construct a super graph $\mathcal{G}$ by contracting each cluster into a single vertex, and (recursively) compute a spanner $\mathcal{H}$ for the super graph $\mathcal{G}$ with stretch $k'<k$. Eventually, for each super edge in $\mathcal{H}$, we will add a representative edge into the resulting spanner $H$. 
The basic insight, is that while the usage of a cluster graph instead of the actual graph increases the stretch by a multiplicative factor of the clusters diameter, we are able to compute a spanner with stretch $k'$ considerably smaller than $k$, and thus somewhat compensating for the loss in the stretch.

This phenomena has opposite effects when applying it on either \cite{BS07} or \cite{KW14} clustering schemes. 
Specifically, applying this idea on \cite{BS07} for $g$ recursive steps, we will obtain stretch $2^g\cdot k$ (compared with $2k-1$ in \cite{BS07}) while reducing the number of passes to $g\cdot k^{\nicefrac1g}$ (compared with $k$ in \cite{BS07}). That is we get a polynomial reduction in the number of passes, while paying a constant increase in stretch.
From the other hand, applying this idea on \cite{KW14} for $g$ recursive steps, we will obtain stretch $2^{g\cdot k^{\nicefrac1g}}$ (compared to $2^k-1$ in \cite{KW14}) while reducing the number of passes to $g+1$ (compared with $2$ in \cite{KW14}). Thus for each additional pass, we get an exponential reduction in the stretch.

Interestingly, the idea of recursively constructing a spanner by contracting clusters was found and used concurrently and independently from us by Biswas \etal \cite{BDGMN20} in the context of the massive parallel computation (MPC) model. They applied it only on \cite{BS07} algorithm in order to construct a spanner in small number of rounds.

\section{Spectral Sparsifiers are Spanners}\label{sec:stream}

In this section, we show that spectral sparsifiers can be used to achieve low stretch spanners in one pass over the stream. Our algorithm works as follows: first, given a graph $G=(V,E)$, it generates a (possibly weighted) spectral sparsifier $H$ of $G$, using the sketches which can be stored in $\tilde{O}(n)$ space \cite{DBLP:conf/focs/KapralovLMMS14,DBLP:journals/corr/abs-1903-12150, DBLP:conf/soda/KapralovMMMNST20}. Then, the weights of all edges are set to be equal to $1$. We show that the resulting graph $\widehat{H}$ is a $\tilde{O}(n^{\frac{2}{3}})$-spanner of the original graph.

\thmSpectralSparsAreSpanners*

As \cite{DBLP:conf/focs/KapralovLMMS14} constructed $(1\pm\eps)$-spectral sparsifier with $O(\frac{n}{\epsilon^2})$ edges in a dynamic stream, by fixing $\eps=\frac{1}{18}$, we conclude:
\corSpectralSparsAreSpanners*
\begin{proof}[{\bf Proof of \Cref{thm:SpectralSparsAreSpanners}}]
	By triangle inequality, it is enough to prove that for every edges $(u,v)\in E$, it holds that $d_{\widehat{H}}(u,v)=\tilde{O}(n^{\frac23})$.
	Our proof strategy is as follows: consider a pair of vertices
	$u,v\in V$ such that $d_{\widehat{H}}(u,v)=s$. We will prove that $R_{u,v}^{G}\ge\tilde{\Omega}(\frac{s^{3}}{n^{2}})$. As for every pair of neighboring vertices it holds that  $R_{u,v}^{G}\le 1$, the theorem will follow.
	
	Consider a pair of vertices
	$v,u\in V$ such that $d_{\widehat{H}}(v,u)=s$. We partition $V$
	to sets $A_{0},A_{1},\dots,A_{s}$ where for $i<s$, $A_{i}=\{z\in V\mid d_{\widehat{H}}(v,z)=i\}$ are
	all the vertices at distance $i$ from $v$ in $\widehat{H}$. $A_{s}=\{z\in V\mid d_{\widehat{H}}(v,z)\ge s\}$
	are all the vertices at distance at least $s$.
	Let $W_{i}^{H}=w_H(A_i\times A_{i+1})$ be the total weight in $H$ (the weighted sparsifier) of all the edges between
	$A_{i}$ to $A_{i+1}$. Similarly, set $W_{i}^{G}=w_G(A_i\times A_{i+1})$. We somewhat abused notation here, we treat non-existing edges as having weight $0$, while all the edges in the unweighted graph $G$ have unit weight.
	For simplicity of notation set also $W_{-1}^{H}=W_{-1}^{G}=W_{s}^{H}=W_{s}^{G}=0$.
	Note that while $W_i^H$ denotes the size if a cut in $H$, it does not correspond to a cut in $G$ (as e.g. there might be edges from $A_i$ to $A_{i+2}$). Thus, a priori there should not be a resemblance between $W_i^G$ to $W_i^H$.
	Nevertheless, we show that  $W_{i}^{H}$ approximates $W_{i}^{G}$. However, the approximation will depend also on $W_{i-1}^{H},W_{i+1}^{H}$ rather than only on $W_{i}^{H}$.
	\begin{claim}
		\label{claim:CutHG}For every $i$, $W_{i}^{H}-\epsilon\cdot(W_{i-1}^{H}+W_{i}^H+W_{i+1}^{H})\le W_{i}^{G}\le W_{i}^{H}+\epsilon\cdot(W_{i-1}^{H}+W_{i}^{H}+W_{i+1}^{H})$.
	\end{claim}
\begin{proof}[Proof of \Cref{claim:CutHG}]
	For a fixed $i$, set
	\[
	\renewcommand{\arraystretch}{2}
	\begin{array}{ll}
	A_{<i}=A_{0}\cup\dots\cup A_{i-1} \qquad\qquad\qquad& A_{>i+1}=A_{i+2}\cup\dots\cup A_{s}\\
	A_{\le i}=A_{0}\cup\dots\cup A_{i} \qquad\qquad\qquad& A_{\ge i+1}=A_{i+1}\cup\dots\cup A_{s}
	\end{array}
	\]
	In addition we denote the weight of several edge sets as follows, (see \Cref{fig:CutWith4sets} for illustration)
	\begin{equation}
	\renewcommand{\arraystretch}{2}
	\begin{array}{lll}
	a=w_{G}(A_{i}\times A_{i+1}) \qquad\qquad\qquad& b=w_{G}(A_{i}\times A_{>i+1})\qquad\qquad\qquad& c=w_{G}(A_{<i}\times A_{i+1})\\
	d=w_{G}(A_{<i}\times A_{>i+1}) \qquad\qquad\qquad& e=w_{G}(A_{<i}\times A_{i}) \qquad\qquad\qquad& f=w_{G}(A_{i+1}\times A_{>i+1})
	\end{array}\label{eq:Defabcdef}
	\end{equation}
	\begin{figure}[t]
		\centering{\includegraphics[scale=.65]{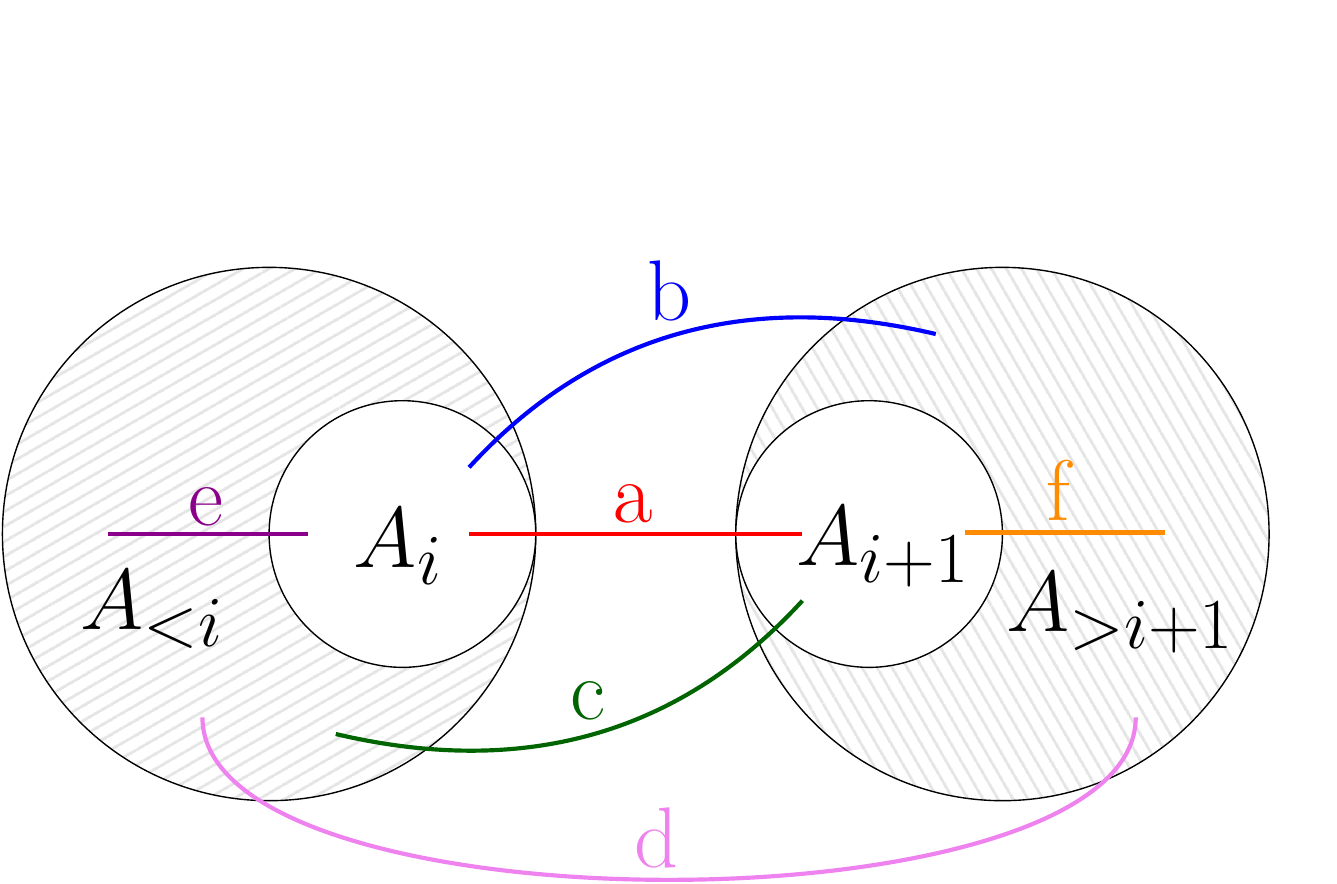}}
		\caption{\label{fig:CutWith4sets}\small An illustration of the diffferent edges sets, the weight of which is denoted in \Cref{eq:Defabcdef}. Note that $a=W_i^G$, $e=W_{i-1}^G$, and $f=W_{i+1}^G$.}
	\end{figure}
	Similarly by replacing $w_{G}$ with $w_{H}$ in \Cref{eq:Defabcdef}, we obtain the values $a',b',c',d',e',f'$ (e.g. $a'=w_{H}(A_{i}\times A_{i+1})$). 
	Note that by the definition of the sets $A_{0},\dots,A_{s}$, it holds that $b'=c'=d'=0$. 
	Using this notation, \Cref{claim:CutHG} states that $a'-\epsilon\cdot(a'+e'+f')\le a\le a'+\epsilon\cdot(a'+e'+f')$.

	Note that any $(1\pm \eps)$-spectral sparsifier is a $(1\pm \eps)$-cut sparsifier. Thus, as $H$ is a $(1\pm \eps)$-spectral sparsifier of $G$, it preserves weights of
	all the cuts up to $\epsilon$ error factors. We derive the following inequalities:
	\begin{tabular}{lclcll}
		$(1-\epsilon)a'$       & $\le$ & $a+b+c+d$ & $\le$ & $(1+\epsilon)a'$      &By $(A_{\le i},A_{\ge i+1})$-cut   \\
		$(1-\epsilon)f'$       & $\le$ & $b+d+f$   & $\le$ & $(1+\epsilon)f'$       &By $(A_{\le i+1},A_{> i+1})$-cut  \\
		$(1-\epsilon)e'$       & $\le$ & $c+d+e$   & $\le$ & $(1+\epsilon)e'$   &By $(A_{<i},A_{\ge i})$-cut      \\
		$(1-\epsilon)(a'+e'+f')$ & $\le$ & $a+d+e+f$ & $\le$ & $(1+\epsilon)(a'+e'+f')$ &By $(A_{<i}\cup A_{i+1},A_i\cup A_{>i+1})$-cut
	\end{tabular}
	
	Or equivalently
	\[
	\begin{array}{ccccc}
	(1-\epsilon)a' & \le & a+b+c+d & \le & (1+\epsilon)a'\\
	-(1+\epsilon)f' & \le & -b-d-f & \le & -(1-\epsilon)f'\\
	-(1+\epsilon)e' & \le & -c-d-e & \le & -(1-\epsilon)e'\\
	(1-\epsilon)(a'+e'+f') & \le & a+d+e+f & \le & (1+\epsilon)(a'+e'+f')
	\end{array}
	\]
	By summing up these 4 inequalities, and dividing by $2$, we get
	\[
	a'-\epsilon\cdot(a'+e'+f')\le a\le a'+\epsilon\cdot(a'+e'+f')~.
	\]
	The claim now follows.
\end{proof}

	Our next goal is to bound $\sum_{i=0}^{s-1}\frac{1}{W_{i}^{H}}$, as this quantity lower-bounds the resistance between $u$ and $v$ in $H$. Since $\sum_{i=0}^{s}|A_i|=n$ and $W_{i}^{G}\le |A_i|\cdot|A_{i+1}|$, one can bound $\sum_{i=0}^{s-1}\frac{1}{W_{i}^{G}}$  by $\Omega\left(\frac{s^{3}}{n^{2}}\right)$. However relating this quantity to the effective resistances in $G$ is not as straightforward as one might expect.
		\begin{claim}
		\label{claim:sumWH}$\sum_{i=0}^{s-1}\frac{1}{W_{i}^{H}}\ge\Omega\left(\frac{s^{3}}{n^{2}}\cdot\frac{\log^{2}\frac{1}{\epsilon}}{\log^{2}n}\right)$.
	\end{claim}
	\begin{proof}[Proof of \Cref{claim:sumWH}]
		For all $i\in[s]$, set $a_i=|A_i|$.
		Set
		\begin{align}\label{eq:alpha}
		\alpha:=10\log_{\frac{1}{6\epsilon}}n^{2}~,
		\end{align}
		and $$I:=\left\{ i\in[s]\mid a_{i}\le\frac{\alpha n}{s}\right\}.$$
		It holds that $|I|\ge\left(1-\frac{1}{\alpha}\right)s+1$, as otherwise
		there are at least  $\frac{s}{\alpha}$ indices $i$ for which $a_{i}>\frac{\alpha n}{s}$,
		implying $\sum_{i}a_{i}>n$, a contradiction, since $A_{0},\dots,A_{s}$
		forms a partition of $V$. Set $$\wt{I}:=\left\{ i\mid\text{such that }\forall j~\text{such that }|i-j|\le\frac{\alpha}{10},~\text{ it holds that }j\in I\right\}. $$
		Note that, since there are less than $\frac{s}{\alpha}$ indices $i$ such that $i\notin I$, then there are less than $\frac{s}{\alpha}\cdot\frac{2\alpha}{10}\le\frac{s}{5}$
		indices out of $\wt{I}$, implying 
		\begin{align}\label{eq:Itilde}
		\left|\wt{I}\right|\ge\frac{s}{2}~.
		\end{align}
		
		Fix an index $i_0\in\wt{I}$, we argue that $W_{i_0}^{H}\le2\left(\frac{\alpha n}{s}\right)^{2}$. 
		For every index $j\in \left[i_0-\frac{\alpha}{10},i_0+\frac{\alpha}{10}-1\right]$, it holds that $W_{j}^{G}=|E_G(A_j,A_{j+1})|\le a_{j}\cdot a_{j+1}\le\left(\frac{\alpha n}{s}\right)^{2}$. 
		Assume for the sake of contradiction that $W_{i_0}^{H}>2\left(\frac{\alpha n}{s}\right)^{2}$.
		We prove by induction that for $1\le j\le\frac{\alpha}{10}$, there is an index $i_j$ such that $|i_j-i_0|\le j$ and $W_{i_{j}}^{H}>\frac{1}{(6\epsilon)^{j}}\left(\frac{\alpha n}{s}\right)^{2}$.
		For the base case, by \Cref{claim:CutHG},
		\[
		W_{i_{0}-1}^{H}+W_{i_{0}}^{H}+W_{i_{0}+1}^{H}\ge\frac{1}{\epsilon}\left(W_{i_{0}}^{H}-W_{i_{0}}^{G}\right)>\frac{1}{\epsilon}\left(2\left(\frac{\alpha n}{s}\right)^{2}-\left(\frac{\alpha n}{s}\right)^{2}\right)=\frac{1}{\epsilon}\left(\frac{\alpha n}{s}\right)^{2}~.
		\]
		The we can choose $i_{1}\in\left\{ i_0-1,i_0,i_0+1\right\} $ such that $W_{i_{1}}^{H}>\frac{1}{3\epsilon}\left(\frac{\alpha n}{s}\right)^{2}>\frac{1}{6\epsilon}\left(\frac{\alpha n}{s}\right)^{2}$.\\
		For the induction step, suppose that there is an index $i_j$ such that $|i_j-i_0|\le j<\frac{\alpha}{10}$ and ${W_{i_{j}}^{H}>\frac{1}{(6\epsilon)^{j}}\left(\frac{\alpha n}{s}\right)^{2}}$.
		As $|i_j-i_0|\le \frac{\alpha}{10}-1$, it follows that  $W_{i_j}^{G}\le\left(\frac{\alpha n}{s}\right)^{2}$. Hence
		\[
		W_{i_{j}-1}^{H}+W_{i_{j}}^{H}+W_{i_{j}+1}^{H}\ge\frac{1}{\epsilon}\left(W_{i_{j}}^{H}-W_{i_{j}}^{G}\right)\ge\frac{1}{\epsilon}\left(\frac{1}{(6\epsilon)^{j}}\left(\frac{\alpha n}{s}\right)^{2}-\left(\frac{\alpha n}{s}\right)^{2}\right)>\frac{1}{2\epsilon}\cdot\frac{1}{(6\epsilon)^{j}}\left(\frac{\alpha n}{s}\right)^{2}~.
		\]
		Thus there is an index $i_{j+1}\in\left\{ i_{j}-1,i_j,i_{j}+1\right\}$ such that $W_{i_{j+1}}^{H}>\frac{1}{(6\epsilon)^{j+1}}\left(\frac{\alpha n}{s}\right)^{2}$,
		as required.
		
		We conclude that,
		\begin{align*}
		W_{i_{\frac{\alpha}{10}}}^{H} & ~>~\left(6\epsilon\right)^{-\frac{\alpha}{10}}\left(\frac{\alpha n}{s}\right)^{2}~\overset{(\ref{eq:alpha})}{\ge}~ n^{2}\left(\frac{\alpha n}{s}\right)^{2}~\ge~ n^{2}~,
		\end{align*}
		where the last inequality follows as $s\le n$. This is a contradiction, as $H$
		is an $(1\pm\epsilon)$ spectral sparsifier of the unweighted graph $G$,
		where the maximal size of a cut is $\frac{n^{2}}{4}$. We conclude
		that for every $i\in\wt{I}$, it holds that $W_{i}^{H}\le2\left(\frac{\alpha n}{s}\right)^{2}$.
		The claim now follows as 
		\begin{align}
		\sum_{i=0}^{s-1}\frac{1}{W_{i}^{H}}&\ge\left|\wt{I}\right|\cdot\frac{1}{2}\left(\frac{\alpha n}{s}\right)^{-2}\nonumber\\
		&\ge\frac{s^{3}}{4\alpha^{2}n^{2}}&&\text{By \Cref{eq:Itilde}}\nonumber\\
		&=\Omega\left(\frac{s^{3}}{n^{2}}\cdot\frac{\log^{2}\frac{1}{\epsilon}}{\log^{2}n}\right)\label{eq:WHsum}&&\text{By \Cref{eq:alpha}}
		\end{align}
	\end{proof}
	
We are now ready to prove the theorem. Construct an auxiliary graph $H'$ from $H$, by  contracting all the vertices inside each set $A_{i}$, and keeping multiple edges. Note that by this operation, the effective resistance between $u$ and $v$ cannot increase.
The graph $H'$ is a path graph consisting of $s$ vertices, where the conductance between the $i$'th vertex to the $i+1$'th is $W_{i}^H$. Using \Cref{claim:sumWH}, we conclude
\begin{align}
	(1+\epsilon)R_{u,v}^{G} & \ge  R_{u,v}^{H}&&\text{By \Cref{fact:EffSparse}}\nonumber\\
	&\ge R_{u,v}^{H'}&&\text{As explained above}\nonumber\\
	&=\sum_{i=0}^{s-1}\frac{1}{W_{i}^{H}}&&\text{Since $H'$ is a path graph}\nonumber\\
	&=\Omega\left(\frac{s^{3}}{n^{2}}\cdot\frac{\log^{2}\frac{1}{\epsilon}}{\log^{2}n}\right)~\label{eq:RGbound}&&\text{By \Cref{eq:WHsum}}
\end{align}
As $u,v$ are neighbors in the unweighted graph $G$, it necessarily holds that $R_{u,v}^{G}\le 1$, implying that $s=O\left(\left(n^{2}\cdot\frac{\log^{2}n}{\log^{2}\frac{1}{\epsilon}}\right)^{\frac{1}{3}}\right)=\wt{O}\left(n^{\frac{2}{3}}\right)$.
\end{proof}

We state the following corollary, based on the last part of the proof of \Cref{thm:SpectralSparsAreSpanners}. 
\begin{corollary}\label{cor:stretcheffres}
	Let $G=(V,E)$ be an unweighted undirected graph, and let $H$ be a $(1\pm\epsilon)$-spectral sparsifier of $G$ for some small enough constant $\epsilon$. Also, let $\widehat{H}$ denote the unweighted $H$. If for a pair of vertices $u,v\in V$ we have $s:=d_{\widehat{H}}(u,v)$, then $$R_{u,v}^{G} = \wt{\Omega}\left(\frac{s^3}{n^2}\right),$$
	and 
	$$R_{u,v}^{H} = \wt{\Omega}\left(\frac{s^3}{n^2}\right).$$
\end{corollary}
\subsection{Sparse graphs}

Suppose we are guaranteed that the graph $G$ we receive in the dynamic stream has eventually at most $m$ edges. In \Cref{thm:SpectralSparsAreSpannersEdges} we show that the distortion guarantee of a sparsifier is at most $\wt{O}(\sqrt{m})$, and thus together with \Cref{thm:SpectralSparsAreSpanners} it is $\wt{O}(\min\{\sqrt{m},n^{\nicefrac{2}{3}}\})$. Later, in \Cref {sec:CongestedMultiplayerModel} we will use this to obtain a two pass algorithm in the simultaneous communication model with distortion $\wt{O}(n^{\nicefrac35})$.

\begin{theorem}\label{thm:SpectralSparsAreSpannersEdges}
	Let $G=(V,E)$ be an undirected, unweighted such that $|V|=n$ and $|E|=m$. For a parameter $\eps\in(0,\frac{1}{18}]$, suppose that $H$ is a $(1\pm\eps)$-spectral sparsifier of $G$. Then $\widehat{H}$ is an $\wt{O}(\sqrt{m})$-spanner of $G$, where $\widehat{H}$ is the unweighted version of $H$.
\end{theorem}
The proof follows similar lines to the proof of \Cref{thm:SpectralSparsAreSpanners} and is deferred to \Cref{sec:omitted-3}. \Cref{thm:SpectralSparsAreSpannersEdges} implies a streaming algorithm using space $\tilde{O}(n)$ that constructs a spanner with stretch $\tilde{O}(\sqrt{m})$. Notice that the number of edges $m$, does not need to be known in advance.

Similar to \Cref{cor:stretcheffres}, using the last part of the proof of \Cref{thm:SpectralSparsAreSpannersEdges}, we conclude the following:
\begin{corollary}\label{cor:stretcheffresedge}
	Let $G=(V,E)$ be an unweighted undirected graph with $m=|E|$, and let $H$ be a $(1\pm\epsilon)$-spectral sparsifier of $G$ for some small enough constant $\epsilon$. Also, let $\widehat{H}$ denote the unweighted $H$. If for a pair of vertices $u,v\in V$ we have $s:=d_{\widehat{H}}(u,v)$, then $$R_{u,v}^{G} = \wt{\Omega}\left(\frac{s^2}{m}\right),$$
	and 
	$$R_{u,v}^{H} = \wt{\Omega}\left(\frac{s^2}{m}\right).$$
\end{corollary}

\subsection{Tightness of \Cref{thm:SpectralSparsAreSpanners} and \Cref{thm:SpectralSparsAreSpannersEdges}}\label{subsec:LowerBound}
In this section, we show that the stretch guarantees in \Cref{thm:SpectralSparsAreSpanners} and \Cref{thm:SpectralSparsAreSpannersEdges} are tight up to polylogarithmic factors. 
\begin{lemma}[Tightness of \Cref{thm:SpectralSparsAreSpanners}]\label{lem:tightnessn}
	For every large enough $n$, there exists an unweighted $n$ vertex graph $G$, and a spectral sparsifier $H$ of $G$ such that $\widehat{H}$ has stretch $\wt{\Omega}(n^{2/3})$ w.r.t. $G$.
\end{lemma}
\begin{proof}
	As was shown by Spielman and Srivastava \cite{DBLP:conf/stoc/SpielmanS08}, one can create a sparsifier $H$ of $G$ (with high probability) by adding each edge $e$ of $G$ to $H$ with probability $p_e=\min\{\epsilon^{-2}\cdot R_e^G \cdot \log n,1\}$ (and weight $1/p_e$).
	This approach is known as \emph{spectral sparsification using effective resistance sampling}. We will construct a graph $G$ and argue that for a random graph $H$ sampled according to the scheme above \cite{DBLP:conf/stoc/SpielmanS08}, the stretch of $\widehat{H}$ will (likely) be $\wt{\Omega}(n^{2/3})$.
	
	For brevity, we will construct a graph with $n+2$ vertices and ignore rounding issues. The graph $G=(V,E)$ is constructed as follows. Let $N:=\frac{n^{\nicefrac23}}{c}$ for $c:=\log n$. We partition the set of vertices, $V$, into $V_0,V_1,\ldots,V_N,V_{N+1}$, where for each $i\in[1,N]$, we have $|V_i| = a=c n^{\nicefrac{1}{3}}$, and $V_0=\{u\}$, $V_{N+1}=\{v\}$ are singletons. For every $i\in[0,N]$, we connect all vertices in $V_i$ to all vertices in $V_{i+1}$, and furthermore, we connect $u$ and $v$ by an edge called $e$. That is,
	\begin{align*}
	E=\left(\cup_{i=0}^{N} V_i \times V_{i+1}\right) \cup \{(u,v)\}.
	\end{align*}
	\begin{figure}[t]
		\centering{\includegraphics[scale=.65]{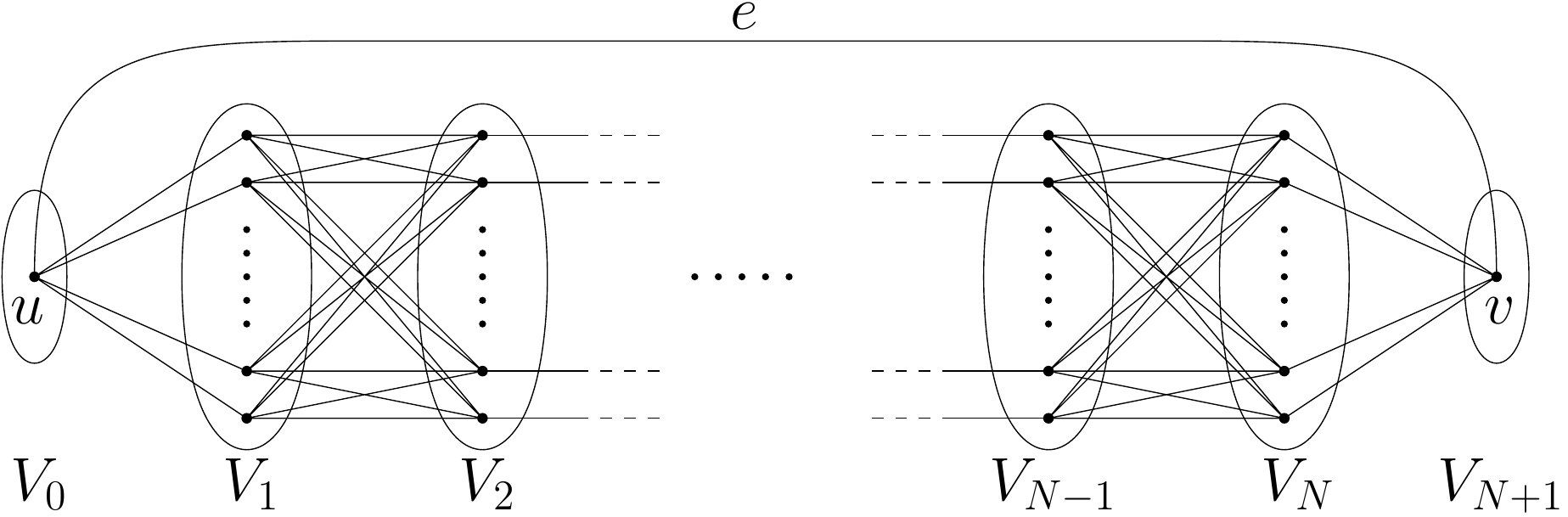}}
		\caption{\label{fig:LB}\small An illustration of the graph $G$ constructed during the proof of \Cref{lem:tightnessn}.}
	\end{figure}	 
	See \Cref{fig:LB} for illustration.
	Next, we calculate $R_e^{G}$, by observing the flow vector when one units of flow is injected in $v$ and is removed from $u$. Denote $R:=R_{e}^{G}$. Then $R$ units of flow is routed using edge $e$, while $(1-R)$ units of flow is routed using the rest of the graph.
	By symmetry, for each cut $V_i\times V_{i+1}$ the flow will spread equally among the edges. Farther, the potential of all the vertices in each set $V_i$ is equal. Denote by $P_i$ the potential of vertices in $V_i$. Thus $0=P_0<P_1<\dots<P_{N+1}=R$.
	For $i=0$, each edge in $V_0\times V_{1}$ carries $\frac{(1-R)}{a}$ flow, thus $P_1-P_0=\frac{(1-R)}{a}$. Similarly, $P_{N+1}-P_N=\frac{(1-R)}{a}$. 
	On the other hand, for $i\in[1,N-1]$, each edge in $V_{i+1}\times V_{i}$ carries $\frac{(1-R)}{a^2}$ flow, thus $P_{i+1}-P_i=\frac{(1-R)}{a^2}$.
	We conclude
	\[
	R=P_{N+1}-P_{0}=\sum_{i=0}^{N}(P_{i+1}-P_{i})=2\cdot\frac{(1-R)}{a}+\frac{(1-R)}{a^{2}}\cdot(N-1)=(1-R)\cdot\frac{2a+|N|-1}{a^{2}}
	\]
	Thus,
	\[
	R_{e}^{G}=R=\frac{2a+|N|-1}{a^{2}-2a-|N|+1}=\frac{2cn^{\nicefrac{1}{3}}+\frac{n^{\nicefrac{2}{3}}}{c}-1}{c^{2}n^{\nicefrac{2}{3}}-2cn^{\nicefrac{1}{3}}-\frac{n^{\nicefrac{2}{3}}}{c}+1}=\frac{1}{c^{3}}(1+o(1))=O(\frac{1}{\log^3 n})~.
	\]
	Note that it is thus most likely that $e$ will not belong to $H$ (for large enough $n$).  For a sampled graph $H$ excluding $e$, we will have $d_{\widehat{H}}(u,v)\ge |N|=\wt{\Omega}(n^{2/3})$.
	From the other hand, as a graph $H$ sampled in this manner is a spectral sparsifier with high probability, it implies the existence of a spectral sparsifier $H$ of $G$ with stretch $\wt{\Omega}(n^{2/3})$, as required.
\end{proof}

\begin{lemma}[Tightness of \Cref{thm:SpectralSparsAreSpannersEdges}]	
	For every large enough $m$, there exists an unweighted graph $G$ with $m$ edges, and a spectral sparsifier $H$ of $G$ such that $\widehat{H}$ has stretch $\wt{\Omega}(\sqrt{m})$ w.r.t. $G$.
\end{lemma}
\begin{proof}
	Fix $n=(\frac{m}{2\log m})^{\nicefrac{3}{4}}$. Note that the graph we constructed during the proof of \Cref{lem:tightnessn} has $2a+(N-1)a^{2}=2cn^{\nicefrac{1}{3}}+(\frac{n^{\nicefrac{2}{3}}}{c}-1)\cdot c^{2}n^{\nicefrac{2}{3}}<2c\cdot n^{\nicefrac{4}{3}}<m$
	edges.
	We can complement it to exactly $m$ edges by adding some isolated component.
	Following \Cref{lem:tightnessn}, this graph has a sparsifier $H$, such that $\widehat{H}$ has stretch $\wt{\Omega}(n^{2/3})=\wt{\Omega}(\sqrt{m})$ w.r.t. $G$, as required.
	
\end{proof}

\begin{remark}\label{rem:CutSparsifier}
	Cut sparsifiers are somewhat weaker version of spectral sparsifiers. Specifically, a weighted subgraph $H$ of $G$ is called a cut sparsifier if it preserves the size of all cuts (up to $1\pm\eps$ factor).
	A natural question is the following: given a cut sparsifier $H$ of $G$, how good of a spanner is $\widehat{H}$?\\
	The answer is: very bad. Specifically, consider the hard instance constructed during the proof of \Cref{lem:tightnessn}. Construct the same graph $G$ where we change the parameter $N$ to equal $\Theta(\frac{n}{\log n})$ and $a$ to $\Theta(\log n)$. 
	There exist a cut sparsifier $H$ of $G$ excluding the edge $e=(u,v)$.
	In particular, $\widehat{H}$ will have stretch $\widetilde{\Omega}(n)$.
\end{remark}

\subsection{Stretch-Space trade-off}\label{sec:streamingtradeoff}
In this section, we first prove a result, which given an algorithm that uses $\wt{O}(n)$ space in the dynamic streaming setting, converts it to an algorithm that uses $\wt{O}(n^{1+\alpha})$ space and achieves a better stretch guarantee (see \Cref{thm:meta-space-stretch}). Then, we apply this theorem to \Cref{cor:SpannerBySparsifiers} and get a space-stretch trade off. Next, in \Cref{thm:TradeOffOnePassEdges} we prove a similar trade off in terms of number of edges. 
\begin{theorem}\label{thm:meta-space-stretch}
	 Assume there is an algorithm, called \textsc{Alg}, that given a graph $G=(V,E)$ in a dynamic stream, with $|V|=n$, using $\wt{O}(n)$ space, outputs a spanner with stretch $\wt{O}(n^{\beta})$ for some constant $\beta \in (0,1)$ with failure probability $n^{-c}$ for some constant $c$. Then, for any constant $\alpha\in(0,1)$, one can construct an algorithm that uses $\wt{O}(n^{1+\alpha})$ space and outputs a spanner with stretch $\wt{O}(n^{\beta(1-\alpha)})$ with failure probability  $\tilde{O}(n^{(2+c)\alpha-c})$.
\end{theorem}
\begin{proof}
	Let $\mathcal{P}\subset 2^{[n]}$ be a set of subsets of $[n]$ such that: (1) $|\mathcal{P}|=O(n^{2\alpha}\log n)$, (2) every $P\in\mathcal{P}$ is of size $|P|=O(n^{1-\alpha})$, and (3) for every $i,j\in[n]$ there is a set $P\in\mathcal{P}$ containing both $i,j$. Such a collection $\mathcal{P}$ can be constructed by a random sampling.	
	Denote $V=\{v_1,\dots,v_n\}$.
	For each $P\in\mathcal{P}$, set $A_P=\{v_i\mid i\in P\}$. 
	For each $P\in\mathcal{P}$, we use \textsc{Alg} independently to construct a spanner $H_P$ for $G[A_P]$ the induced graph on $A_P$. The final spanner will be their union $H=\cup_{P\in\mathcal{P}} H_P$.
	
	The space (and also the number of edges in $H$) used by our algorithm is bounded by
	$\sum_{P\in\mathcal{P}}\wt{O}(|P|)=\wt{O}(n^{2\alpha}\cdot n^{1-\alpha})=\wt{O}(n^{1+\alpha})$.
	From the other hand, for every $v_i,v_j\in V$ such that $i,j\in P$, it holds that 
	\[
	d_{H}(v_{i},v_{j})\le d_{H_{P}}(v_{i},v_{j})\le\wt{O}(|P|^{\beta})\le\wt{O}(n^{\beta(1-\alpha)})~.
	\]

	By union bound, the failure probability is bounded by 
	$\tilde{O}(n^{2\alpha})\cdot O(n^{-c(1-\alpha)})=\tilde{O}(n^{(2+c)\alpha-c})$.
	
\end{proof}

Combining  \Cref{cor:SpannerBySparsifiers} with  \Cref{thm:meta-space-stretch}, we conclude:
\thmOnePassTradeoff*
\begin{remark}
	We can reduce the number of edges in the spanner returned to $O(n)$, by incurring additional $O(\log n)$ factor to the stretch. This is done by computing additional spanner upon the one returned by \Cref{thm:OnePassTradeoff}.
\end{remark}

	Following the approach in \Cref{thm:meta-space-stretch}, we can also use more space to reduce the stretch parameterized by the number of edges.
	Note that the \Cref{thm:TradeOffOnePassEdges}  provides better result than \Cref{thm:OnePassTradeoff} when $m\le n^{\frac43+\frac23\alpha}$.
\begin{theorem}\label{thm:TradeOffOnePassEdges}
	Consider an $n$-vertex unweighted graph $G$, the edges of which arrive in a dynamic stream. 
	For every parameter $\alpha\in(0,1)$, there is an algorithm using $\wt{O}(n^{1+\alpha})$ space, constructs a spanner with stretch $\wt{O}(\sqrt{m}\cdot n^{-\alpha})$.
\end{theorem}
\begin{proof}
	Similarly to \Cref{thm:meta-space-stretch}, our goal here is to partition the vertices into $\approx n^{2\alpha}$ sets of similar size. However, while in \Cref{thm:meta-space-stretch} we wanted to bound the number of vertices in each set, here we want to bound the edges in each set. As the edge set is unknown, we cannot use a fixed partition. Rather, in the preprocessing phase we will sample a partition that  w.h.p. will be good w.r.t. arbitrary fixed edge set.
	
	Fix $p=n^{-\alpha}$. With no regard to the rest of the algorithm, during the stream we will sample $\widetilde{O}(n^{1+\alpha})=\widetilde{O}(\frac np)$ edges from the stream using sparse recovery (\Cref{lem:RecoverSingleEdge}), and add them to our spanner $\widehat{H}$. If $m\le np^{-1}$, we will restore the entire graph $G$, and thus will have stretch $1$. The rest of the analysis will be under the assumption that $m>np^{-1}$.
	
	For every $i\in [1, \frac{8}{p^{2}}\ln n]$, sample a subset $A_i$ by adding each vertex with probability $p$.
	Consider a single subset $A_i$ sampled in this manner, and denote $G_i=G[A_i]$ the graph it induces. We will compute a sparsifier $H_i$ for $G_i$. Our final spanner will be $\widehat{H}=\cup_i \widehat{H}_i$ a union of the unweighted versions of all the sparsifiers (in addition to the random edges sampled above). The space we used for the algorithm is $\sum_{i}\widetilde{O}(|A_{i}|)$. 
	Note that with high probability, by Chernoff inequality $\sum_{i}\widetilde{O}(|A_{i}|)=\widetilde{O}(n^{2\alpha}\cdot n^{1-\alpha})=\widetilde{O}(n^{1+\alpha})$.
	
	Next we bound the stretch. Consider a pair of vertices $(u,v)\in E$. 
	Denote by $\psi_i$ the event that both $u,v$ belong to $A_i$. Note that $\Pr[\psi_i]=p^2$. Denote by $m_i=\left|{A_{i} \choose 2}\cap E\right|$ the number of edges in $G_i$.
	Set 
	\[
	\mu_{i}=\mathbb{E}\left[m_i\mid\psi_i\right]\le1+p\cdot\left(\deg_{G}(v)+\deg_{G}(u)\right)+mp^{2}<1+2np+mp^{2}<4mp^{2}~,
	\]
	to be the expected number of edges in $G_i$ provided that $u,v\in A$. The first inequality follows as (1) $(u,v)\in G_i$, (2) every edge incident on $u,v$ belongs to $G_i$ with probability $p$, and (3) every other edge belongs to $G_i$ with probability $p^2$.
	In the final inequality we used the assumption $n<mp$.
	Denote by $\phi_i$ the event that $m_i\le8mp^{2}$. By Markov we have 
	\[
	\Pr\left[\psi_{i}\wedge\phi_{i}\right]=\Pr\left[\psi_{i}\right]\cdot\Pr\left[\phi_{i}\mid\psi_{i}\right]\ge\frac{1}{2}p^{2}~.
	\]
	As $\{\psi_{i}\wedge\phi_{i}\}_i$ are independent, we have that the probability that none of them occur is bounded by
	\[
	\Pr\left[\bigwedge_{i}\left(\overline{\psi_{i}\wedge\phi_{i}}\right)\right]\le(1-\frac{1}{2}p^{2})^{\frac{8}{p^{2}}\ln n}<e^{-\frac{1}{2}p^{2}\cdot\frac{8}{p^{2}}\ln n}=n^{-4}~.
	\]	
	Note that if both $\psi_{i},\phi_{i}$ occurred, and $H_i$ is an $1\pm\eps$ sparsifier of $G_i$, by \Cref{thm:SpectralSparsAreSpanners} we will have that 
	\[
	d_{\widehat{H}}(u,v)\le d_{\widehat{H}_{i}}(u,v)\le\wt{O}(\sqrt{m_{i}})=\wt{O}(\sqrt{m}\cdot p)=\wt{O}(\sqrt{m}\cdot n^{-\alpha})
	\]
	By union bound, the probability that for every $(u,v)\in E$, there is some $i$ such that $\psi_{i}\wedge\phi_{i}$ occurred is at least $1-n^{-2}$. The probability that every $G_i$ is a spectral sparsifier is $1-n^{-\Omega(1)}$. The theorem follows by union bound.

\end{proof}

\section{Simultaneous Communication Model}\label{sec:CongestedMultiplayerModel}	
	In \Cref{sec:stream}, we considered streaming model and proved results for the setting when one pass over the stream was allowed. The remaining question is as follows: using small number of communication rounds (but more than $1$), can we improve the stretch of a spanner constructed in the simultaneous communication model?
	A partial answer is given in the following subsections. 
	
	First, in \Cref{sec:filtering} we present a single filtering algorithm that provides two different trade-offs between stretch and number of communication rounds (see \Cref{alg:filttering} and \Cref{thm:filtering}).
	Basically, the algorithm receives a parameter $t>1$, in each communication round, an unweighted version of a sparsifier is added to the spanner. Then, locally in each vertex, all the edges that already have a small stretch in the current spanner are deleted (stop being considered), and another round of communication begins.
	
	In \Cref{thm:filtering} we present two arguments. The first argument is based on effective resistance filtering, which results in a spanner with $\tilde{O}(n^{\frac{g+1}{2g+1}})$ stretch in $g$ communication rounds. The second argument, which is based on low-diameter decomposition, results in a spanner with $\tilde{O}\left(n^{\frac{2}{g}}\right)$ stretch in $g$ communication rounds. The latter approach outputs a spanner with smaller stretch compared to the former algorithm for $g\ge 4$.

	Finally, in \Cref{subsec:SimModelTradeoff} we generalize our results to the case where each player is allowed $\tilde{O}(n^\alpha)$ communication per round for some $\alpha\in(0,1)$. In that section, we prove two results: (1) in \Cref{thm:SCmodel-1round} we give a space (communication per player) stretch trade off for one round of communication (2) in \Cref{thm:SCmodel} we give a similar trade off for more than one round of communication.

\subsection{The filtering algorithm}\label{sec:filtering}
	The algorithm will receive a stretch parameter $t$. During the execution of the algorithm, we will hold in each step a spanner $\widehat{H}$, and a subset of unsatisfied edges.
	As the algorithm proceeds, the spanner will grow, while the number of unsatisfied edges will decrease.
	Initially, we start with an empty spanner $\widehat{H}$, and the set of  unsatisfied edges $E_0=E$ is the entire edge set.
	In general, at round $i$, we hold a set $E_i$ of edges yet unsatisfied. We construct a spectral sparsifier $H_i$ for the graph $G_i=(V,E_i)$ over thus edges. $\widehat{H}_i$, the unweighted version of $H_i$ is added to the spanner $\widehat{H}$. $E_{i+1}$ is defined to be all the edges $(u,v)\in E_i$, for which the distance in $\widehat{H}$ is greater than $t$, that is $d_{\widehat{H}}(u,v)>t$. Note that as the sparsifier $H_i$, and hence the spanner $\widehat{H}$ is known to all, each vertex locally can compute which of its edges belong to $E_{i+1}$. 
	
	In addition, the algorithm will receive as an input parameter $g$ to bound the number of communication rounds. 
	We denote by $E_{g+1}$ the set of unsatisfied edge by the end of the algorithm. That is edges from $(u,v)\in E$ for which $d_{\widehat{H}}(u,v)>t$.
	Note that during the execution of the algorithm,
	$E_{g+1}\subseteq E_g\subseteq E_{g-1}\subseteq\cdots\subseteq E_1=E$.	
	Finally, if $E_{g+1}=\emptyset$, it will directly imply that $\widehat{H}$ is a $t$-spanner of $G$.
	See \Cref{alg:filttering} for illustration.
	
\begin{algorithm}[t]
	\caption{\texttt{Spanners Using Filtering}$(G=(V,E),t,g)$}  
	\DontPrintSemicolon
	\SetKwInOut{Input}{input}\SetKwInOut{Output}{output}
	\Input{Graph $G=(V,E)$ (in simultaneous communication model), number of rounds $g$, stretch parameter $t$}	
	\Output{$t$-spanner of $G$ with $\tilde{O}(n\cdot g)$ edges}
	\BlankLine
	\label{alg:filttering}
	$\epsilon \gets \frac{1}{18}$\;
	$\widehat{H} \gets \emptyset$ \tcp*{$\widehat{H}$ will be the output spanner}
	
	\For{$i=1$ to $g$}{
		$E_{i}\gets \{e=(u,v)\in G_{i-1} \text{ such that } d_{\widehat{H}}(u,v)> t\}$\;
		$G_i\gets(V,E_i)$\;
		Let $H_i$ be a $(1\pm \epsilon)$-spectral sparsifier of graph $G_i$. \;
		$\widehat{H}\gets \widehat{H} \cup \widehat{H}_i$\tcp*{$\widehat{H}_i$ is the unweighted version of $H_i$}
	}
	\Return $\widehat{H}$
\end{algorithm}
Below, we state the theorem, which proves the round complexity and correctness of \Cref{alg:filttering}.
\thmfiltering*
\begin{proof}
For $g=1$, the theorem holds due to \Cref{thm:SpectralSparsAreSpanners}, thus we will assume that $g\ge2$.
We prove each of the two upper-bounds on stretch separately. We prove the first bound using an effective resistance based argument. The latter upper-bound is proven using an argument based on filtering low-diameter clusters. 
\paragraph{Effective resistance argument:}
	We execute  \Cref{alg:filttering} with parameter $g$, and $t=\tilde{O}(n^{\frac{g+1}{2g+1}})$.
	Consider an edge $e=(u,v)\in E_1$. If $e\in E_{2}$, then it follows from \Cref{cor:stretcheffres} that  $R_{u,v}^{H_{1}}=\tilde{\Omega}\left(\frac{t^{3}}{n^{2}}\right)$.
	Set  $a_1=\tilde{\Omega}\left(\frac{t^{3}}{n^{2}}\right)$. Then 
	\begin{align}
		|E_{2}|\le\frac{1}{a_{1}}\sum_{e\in E_{1}}R_{e}^{H_{1}}\le\frac{1+\eps}{a_{1}}\sum_{e\in E_{i}}R_{e}^{G_{1}}\le\frac{1+\eps}{a_{1}}\cdot(n-1)\le\tilde{\Omega}\left(\frac{n^{3}}{t^{3}}\right)~,\label{eq:resistanceFilt}
	\end{align}
	where the first inequality follows as  $a_{1} \le R_{e}^{H_{1}}$ for $e\in E_2$, the second inequality is by \Cref{fact:EffSparse}, and the third inequity follows by \Cref{fact:sumeff}, as $G_{i-1}$ is unweighted.
	In general, for $i\ge2$, we argue by induction that $|E_i|=\tilde{O}\left(\frac{n^{i+1}}{t^{2i-1}}\right)$. Indeed, consider an edge $e\in E_{i+1}$. Using the induction hypothesis, it follows from \Cref{cor:stretcheffresedge} that 
	\[
	R_{u,v}^{H_{i}}=\tilde{\Omega}\left(\frac{t^{2}}{|E_{i}|}\right)=\tilde{\Omega}\left(\frac{t^{2(i+1)-1}}{n^{i+1}}\right)
	\]
	Set $a_i=\tilde{\Omega}\left(\frac{t^{2(i+1)-1}}{n^{i+1}}\right)$. Using the same arguments as in \Cref{eq:resistanceFilt}, we get
	\[
	|E_{i+1}|\le\frac{1}{a_{i}}\sum_{e\in E_{i}}R_{e}^{H_{i}}\le\frac{1+\eps}{a_{i}}\sum_{e\in E_{i}}R_{e}^{G_{i}}\le\frac{1+\eps}{a_{i}}\cdot(n-1)\le\tilde{O}\left(\frac{n^{(i+1)+1}}{t^{2(i+1)-1}}\right)~.
	\]
	Finally, for every $e\in E_g$, following \Cref{thm:SpectralSparsAreSpannersEdges}, it holds that 
	\[
	d_{\widehat{H}}(u,v)\le d_{\widehat{H}_{g}}(u,v)\le\tilde{O}\left(\sqrt{E_{g}}\right)=\tilde{\Omega}\left(\sqrt{\frac{n^{g+1}}{t^{2g-1}}}\right)\le t~,
	\]
	where the last inequality holds for $t=\tilde{\Omega}(n^{\frac{g+1}{2g+1}})$.
	We conclude that $E_{g+1}=\emptyset$. The theorem follows.

\paragraph{Low diameter decomposition argument:}

Fix $\phi=\frac{1}{3}n^{-\nicefrac2g}$. We will execute \Cref{alg:filttering} with parameter $g$ and $t=\frac{4+o(1)}{\phi}\cdot\ln n$. We argue that for every $i\in[2,g+1]$, $|E_{i+1}|\le 3\phi |E_i|$. As $|E_1|<n^2$, it will follow that $E_{g+1}=\emptyset$, as required.

Consider the unweighted graph $G_i$, and the sparsifier $H_i$ we computed for it. We will cluster $G_i$ based on cut sizes in $H_i$.
The clustering procedure is iterative, where in phase $j$ we holds an induced subgraph $H_{i,j}$ of $H_i$, create a cluster $C_j$, remove it from the graph $H_{i,j}$ to obtain an induced subgraph $H_{i,j+1}$, and continue. The procedure stops once all the vertices are clustered.
Specifically, in phase $j$, we pick an arbitrary unclustered center vertex $v_j\in H_{i,j}$, and create a cluster by growing a ball around $v_j$.
Set $B_{r}=B_{\hat{H}_{i,j}}(v_j,r)$ to be the radius $r$ ball around $v_j$ in the unweighted version of $H_{i,j}$. That is $B_{r+1}=B_r\cup N(B_r)$, where $N(B_r)$ are the neighbors of $B_{r}$ in $H_{i,j}$.
Let $r_j$ be the minimal index $r$ such that 
\begin{align}
\partial_{H_{i,j}}(B_r) < \phi\cdot \mathrm{Vol}_{H_{i,j}}(B_r)~.\label{eq:sparseCut}
\end{align} 
\sloppy Here $\partial_{H_{i,j}}(B_r)$ denotes the total weight of the outgoing edges from $B_r$, while $\mathrm{Vol}_{\hat{H}_{i,j}}(B_r)=\sum_{u\in B_r}\deg_{\hat{H}_{i,j}}(u)$ denotes the sum of the weighted degrees of all the vertices in $B_r$.
Note that while $B_r$ is defined w.r.t. an unweighted graph $\widehat{H}_{i,j}$, $\partial_{H_{i,j}}$ and $\mathrm{Vol}_{H_{i,j}}$ are defined w.r.t. the weighted sparsifier.
For every $r$, it holds that $\mathrm{Vol}_{\hat{H}_{i,j}}(B_{r+1})\ge \mathrm{Vol}_{\hat{H}_{i,j}}(B_r)+\partial_{\hat{H}_{i,j}}(B_r)$.
We argue that $r_j\le 2(1+\epsilon)\cdot{n \choose 2}$.
If $v_j$ is isolated in $H_{i,j}$, then \Cref{eq:sparseCut} holds for $r=0$ and we are note. 
Else, as the minimal weight of an edge in a sparsifier is $1$,\footnote{Since we are producing spectral sparsifiers by effective resistance sampling method using corresponding sketches, each edge $e$ is reweighted by $\frac{1}{p_e}$ where $p_e$ is the probability that edge $e$ is sampled, and hence the weights are at least $1$.} it holds that  $\mathrm{Vol}_{H_{i,j}}(B_{0})=\deg_{H_{i,j}}(v_{j})\ge1$.
We conclude that for $r_j$, the minimal index for which \Cref{eq:sparseCut} holds, we  have that 
\[
\mathrm{Vol}_{H_{i,j}}(B_{r_{j}})\ge(1+\phi)\mathrm{Vol}_{H_{i,j}}(B_{r-1})\ge\cdots\ge(1+\phi)^{r_{j}}\mathrm{Vol}_{H_{i,j}}(B_{0})\ge(1+\phi)^{r_{j}}~,
\]
On the other hand, as $H_i$ is a $(1+\eps)$ spectral sparsifier of an unweighted graph $G_i$, we have
\[
\mathrm{Vol}_{H_{i,j}}(B_{r_j})\le\mathrm{Vol}_{H_{i,j}}(H_{i,j})\le2(1+\epsilon)\cdot|E|\le2(1+\epsilon)\cdot{n \choose 2}~.
\]
Therefore, it must holds that $(1+\phi)^{r_j} \le 2(1+\epsilon){n\choose2}$, which implies
\[
r_j\le\frac{\ln((1+\eps)n^{2})}{\ln(1+\phi)}=\frac{2+o(1)}{\phi}\cdot \ln n~.
\]
We set $C_j=B_{r_j}$ and continue to construct $C_{j+1}$. Overall, we found a partition of the vertex set $V$ into clusters $C_1,C_2,\dots$ such that each cluster satisfies \Cref{eq:sparseCut}, and has (unweighted) diameter at most $\frac{4+o(1)}{\phi}\cdot\ln n= t$. 
In particular, for every edge $e=(u,v)\in E_i$, if $u,v$ are clustered to the same $C_i$, then the distance between them in $\hat{H}$ will be bounded by  $t$. Thus $E_{i+1}$ will be a subset $\partial_{H_{i}}(C_{1},C_{2},\ldots)$, the set of inter-cluster edges.
It holds that
\begin{align} \label{eq:partialvol}
\partial_{H_{i}}(C_{1},C_{2},\ldots)=\sum_{j\ge1}\partial_{H_{i,j}}(C_{j})\le\phi\cdot\sum_{j\ge1}\mathrm{Vol}_{H_{i,j}}(C_{j})\le\phi\cdot\mathrm{Vol}_{H_{i}}(V)~,
\end{align}
where the first inequality holds as each edge counted exactly once. For example the edge $(u,v)\in E(C_a,C_b)$, where $a<b$ counted only at $\partial_{\hat{H}_{i,a}}(C_{a})$. 
Hence,
\begin{align*}
|E_{i+1}|\le\partial_{G_{i}}(C_1,C_2,\ldots) &\le (1+\epsilon) 	\partial_{H_i}(C_1,C_2,\ldots)&&\text{By \Cref{fact:cutsp}}\\
&\le (1+\epsilon)\phi \cdot \mathrm{Vol}_{H_i}(V)&&\text{By \Cref{eq:partialvol}} \\
&\le \frac{\phi(1+\epsilon)}{1-\epsilon} \cdot \mathrm{Vol}_{G_i}(V)&&\text{By \Cref{fact:cutsp}}\\
&= \phi \cdot \left(1+\frac{2\epsilon}{1-\epsilon}\right) \cdot 2|E_i|< 3\phi\cdot |E_i|~.
\end{align*}
\end{proof}

\begin{remark}
	Note that in fact for the low diameter decomposition argument, it is enough to use in \Cref{alg:filttering} cut sparsifiers rather than spectral sparsifiers.
\end{remark}

\subsection{Stretch-Communication trade-off}\label{subsec:SimModelTradeoff}

We note that if more communication per round is allowed, then we can obtain the following.
\thmSCMtradeoffoneround*
\begin{proof}
	We prove the stretch bounds, one by one. 
	\paragraph{Proving $\wt{O}(n^{(1-\alpha)\frac{2}{3}})$:} Basically, the claim follows by \Cref{cor:SpannerBySparsifiers}. More specifically, we work on graphs induces on $O(n^{1-\alpha})$ sized set of vertices. For each such subgraph, we can construct sparsifiers using $O(\polylog(n))$ sized sketches communicated by each vertex involved. Since each vertex is involved in $\wt{O}(n^{\alpha})$ subgraphs, then communication per vertex is $\wt{O}(n^{\alpha})$. And by \Cref{cor:stretcheffres}, the stretch is $\wt{O}(n^{(1-\alpha)\frac{2}{3}})$.
	
	\paragraph{Proving $\wt{O}\left(\sqrt{m} \cdot n^{-\alpha}\right)$:} 
		First, the reader should note that we cannot directly use \Cref{cor:stretcheffresedge} for this part. The reason is that during the proof of \Cref{thm:TradeOffOnePassEdges}, for the special case where $m\le n^{1+\alpha}$, we simply used a sparse recovery procedure to recover the entire graph. 
		However, as the graph $G$ might contain a dense subgraph, sparse recovery is impossible in the simultaneous communication model.
		Instead, we use a procedure, called \emph{peeling low degree vertices}, where using $\wt{O}(n^{\alpha})$ bits of information per vertex, we can partition the vertices into two sets, $V_1$ and $V_2$, where all edges incident on $V_1$ are recovered and minimum degree in $G[V_2]$ is at least $n^{-\alpha}$. We present this procedure in \Cref{alg:peeling} and its guarantees are proved in \Cref{claim:ssparse} below.

		\begin{lemma}[Peeling low-degree vertices]\label{claim:ssparse}
			In a simultaneous communication model, where communication per player is  $\tilde{O}(s)$, there is an algorithm that each vertex can  locally run and output a partition of the vertices into $V_1,V_2$ such that:
			\begin{enumerate}
				\item All the incident edges of $V_1$ are recovered.
				\item The min-degree in the induce graph $G[V_2]$ is at least $s$.
			\end{enumerate}
			Furthermore, the partitions output by all vertices are identical, due to the presence of shared randomness. 
		\end{lemma}
		\begin{proof}
			First, we argue that using $s$-sparse recovery procedure on the neighborhood of vertices, one can find a set $V_1\subseteq V$ such that all the vertices in $V\setminus V_1$ have degree more than $s$. This is done in the following way: each vertex prepares an $s$-sparse recovery sketch for its neighborhood, and in the first round of communication writes its sketch alongside its degree on the board. Then, each vertex runs \Cref{alg:peeling} locally. Note that the output is identical in all vertices since they have access to shared randomness.
			
			Now, we argue the correctness of \Cref{alg:peeling}. First, we let $\textsc{Recover}$ be a $s$-sparse recovery algorithm. More specifically, the following fact holds. 
			\begin{fact}\label{fact:sparse-recovery}
				For any integer $s$, given $S$, a $\wt{O}(s)$-bit sized linear $s$-sparse recovery sketch of a vector $\vec{b}$, such that $\textsc{Support}(\vec{b})\le s$, algorithm $\textsc{Recover}(S)$ outputs the non-zero elements of $\vec{b}$, with high probability.
			\end{fact} 
			Consider the execution of \Cref{alg:peeling}. If  in the beginning there does not exist a low-degree vertex, we are done. Otherwise, there exists a vertex $u$ with degree $\le s$. Now, when we call $\textsc{Recover}(S_u)$ it is guaranteed that the support of the vector is bounded by $s$ (see \Cref{line:support} of \Cref{alg:peeling}). In that case, $\textsc{Recover}(S_u)$ succeeds with high probability.
			Note that in case of success, the output of $\textsc{Recover}(S_u)$ is deterministic- that is depend only the graph and not on the random coins.
			Then, we delete vertex $u$ alongside its incident edges. The sketches for the rest of the graph can be updated accordingly, since the sketches are linear. Thus, we can use the updated sketches in the next round to recover the neighborhood of another low-degree vertex (in the updated graph), without encountering dependency issues (as the series of events we should succeed upon is predetermined). We repeat this procedure until no vertex with degree $\le s$ remains.  Furthermore, we call $\textsc{Recover}$ at most $n$ times per vertex (since we can delete at most $n$ vertices), in total, using union bound, the algorithm succeeds with high probability. \footnote{A similar argument is also given in \cite{DBLP:journals/corr/abs-1903-12165}.}
			\begin{algorithm}[t]
				\caption{\texttt{Low-Degree Peeling}$(\{S_u\}_{u\in V},s)$}  
				\DontPrintSemicolon
				\SetKwInOut{Input}{input}\SetKwInOut{Output}{output}
				\Input{A parameter $s$, linear $s$-sparse recovery sketches (denoted by $S_u$ for each vertex $u$)}	
				\Output{A partition of vertices into two sets, $V_1$ and $V_2$, with the guarantees mentioned in \Cref{claim:ssparse}}
				\BlankLine
				\label{alg:peeling}
				$V_1\gets \emptyset$\;
				$V_2\gets V$\;
				\While{$\exists$ a vertex $u$ with degree $\le s$\label{line:support}}{
					$u \gets$ a vertex with degree $\le s$\tcp*{Using a universal ordering, and degrees in $G[V_2]$}
					$E_u\gets \textsc{Recover}(S_u)$\tcp*{See \Cref{fact:sparse-recovery}} 
					Remove $E_u$ from the sketches and update degrees.\tcp*{Sketches are linear}
					$V_1\gets V_1\cup \{u\}$.\;
					$V_2 \gets V_2 \setminus \{u\}$.
				}
				\Return $(V_1,V_2)$
			\end{algorithm}
		\end{proof}

We use \Cref{alg:peeling} with $s=n^{\alpha}$.
In the same time, we use the algorithm from \Cref{thm:TradeOffOnePassEdges}. That is, partition the vertices into $\widetilde{O}(n^{2\alpha})$ sets such that each vertex belong to each set with probability $n^{-\alpha}$. Than compute a sparsifier $H$ for each set and take their union.
It follows that the total required communication is $\widetilde{O}(n^{\alpha})$ per vertex.
Note that the algorithm of \Cref{thm:TradeOffOnePassEdges} is linear. Hence after using \Cref{claim:ssparse}, we can add all the edges incident on $V_1$ to the spanner, and update the algorithm from \Cref{thm:TradeOffOnePassEdges} accordingly. That is we will use it only on $G[V_2]$.
 
Note that we have $|E(G[V_2])| \ge |V_2|\cdot n^{\alpha}$, and consequently we can use the argument in the proof of \Cref{thm:TradeOffOnePassEdges}. 
In total from one hand we will obtain stretch $1$ on edges incident to $V_1$, and from the other hand, for edges inside $G[V_2]$ we will have stretch of $\wt{O}(\sqrt{|E(G[V_2])|}\cdot n^{\alpha})\le \wt{O}(\sqrt{m}\cdot n^{\alpha})$.
\end{proof}
\thmSCMtradeoff*

\begin{proof}
	We use the same set of subsets of vertices as in \Cref{thm:meta-space-stretch}, i.e., let $\mathcal{P}\subset 2^{[n]}$ be a set of subsets of $[n]$ such that: (1) $|\mathcal{P}|=O(n^{2\alpha}\log n)$, (2) every $P\in\mathcal{P}$ is of size $|P|=O(n^{1-\alpha})$, and (3) for every $i,j\in[n]$ there is a set $P\in\mathcal{P}$ containing both $i,j$. Such a collection $\mathcal{P}$ can be constructed by a random sampling.
	Denote $V=\{v_1,\dots,v_n\}$.
	For each $P\in\mathcal{P}$, set $A_P=\{v_i\mid i\in P\}$. 
	For each $P\in\mathcal{P}$, we use \Cref{alg:filttering} independently on each subgraph. Then, using \Cref{thm:filtering} on each subgraph, since the size of each subgraph is $O(n^{1-\alpha})$ and since for each edge we have a subgraph that this edge is present, the claim holds. 
\end{proof}

\section{Pass-Stretch trade-off: smooth transition}\label{sec:spanner-sketch}
In this section we study the trade-off between the stretch and the number of passes in the semi-streaming model. Our contribution here is a smooth transition between the spanner of \cite{BS07} (\Cref{thm:BS}) and that of \cite{KW14} (\Cref{thm:KW}), achieving a general trade-off between number of passes and stretch (while the space/number of edges is fixed).

\subsection{Previous algorithms}
We will use the clusters created in the algorithms of \cite{BS07} and \cite{KW14} as a black box.  For completeness in \Cref{sec:altBS} and \Cref{sec:KW} we provide the construction and proof of \cite{BS07} and \cite{KW14}, respectively.
The properties of the clustering procedure is described in \Cref{lem:BS_clusters} and \Cref{lem:KW_clustering}. See \Cref{sec:altBS} and \Cref{sec:KW} for a discussion of how exactly they follow.

$\mathcal{P}\subseteq 2^V$ is called a partial partition of $V$  if $\cup\mathcal{P}\subseteq V$, and for every $P,P'\in \mathcal{P}$, $P\cap P'=\emptyset$.
We denote by $B(n,p)$ the binomial distribution, where we have $n$ biased coins, each with probability $p$ for head, and we count the total number of heads.
\begin{restatable}[\cite{KW14} clustering]{lemma}{KWclustering}\label{lem:KW_clustering}
	Given an unweighted, undirected $n$-vertex graph $G=(V,E)$ in a streaming fashion, for every parameters $p\in(0,1]$ and integer $i\le\log_{\frac{1}{p}}n$,	
	there is a $2$ pass algorithm that uses $\tilde{O}(\nicefrac{|V|}{p})$ space, and returns a partial partition  $\mathcal{P}$ of $V$, and a subgraph $H$ (where $|H|=\tilde{O}(\nicefrac{|V|}{p}))$ such that:
	\begin{OneLiners}
		\item $\mathcal{P}$ is known at the end of the first pass, and $|\mathcal{P}|$ is distributed according to $B(|V|,p^i)$.
		\item Each cluster $P\in\mathcal{P}$ has diameter at most $2^{i+1}-2$ w.r.t. $H$.
		\item For every edge $(u,v)$ such that at least one of $u,v$ is not in $\cup\mathcal{P}$, it holds that 
		$d_H(u,v)\le 2^i-1$.
	\end{OneLiners}
\end{restatable}
\begin{restatable}[\cite{BS07} clustering]{lemma}{BSclustering}\label{lem:BS_clusters}
	 Given an unweighted, undirected $n$-vertex graph $G=(V,E)$ in a streaming fashion, for every parameters $p\in(0,1]$ and integer $i\le\log_{\frac{1}{p}}n$,	
	there is an $i+1$ pass algorithm that uses $\tilde{O}(\nicefrac{|V|}{p})$ space, and returns a partial partition  $\mathcal{P}$ of $V$, and a subgraph $H$ (where $|H|=\tilde{O}(\nicefrac{|V|}{p}))$ such that:
	\begin{OneLiners}
		\item $\mathcal{P}$ is known at the end of the $i$`th pass, and $|\mathcal{P}|$ is distributed according to $B(|V|,p^i)$.
		\item Each cluster $P\in\mathcal{P}$ has diameter at most $2i$ w.r.t. $H$.
		\item For every edge $(u,v)$ such that at least one of $u,v$ is not in $\cup\mathcal{P}$, it holds that 
		$d_H(u,v)\le 2i-1$.
	\end{OneLiners}
\end{restatable}

\subsection{Algorithms}
We begin with  a construction based on \Cref{lem:KW_clustering}. In this regime we are interested in at most $\log k$ passes.

\thmKWInPasses*
\begin{proof}	
Fix $r=\left\lceil (\frac{k+1}{2})^{\nicefrac{1}{g}}\right\rceil -1$.
For $i\in[1,g+1]$ set $d_{1}=\frac{1}{k}$ and in general $d_{i}=\frac{1}{k}+r\sum_{q=1}^{i-1}d_{q}$.
By induction it holds that $d_{i}=\frac{(r+1)^{i-1}}{k}$, as 
\begin{align*}
d_{s+1} & =\frac{1}{k}+r\sum_{q=1}^{s}d_{q}=r\cdot d_{s}+\frac{1}{k}+r\sum_{q=1}^{s-1}d_{q}=(r+1)\cdot d_{s}=\frac{(r+1)^{s}}{k}~.
\end{align*}
In addition, set $p_{i}=n^{-d_{i}}$. 
Our algorithm will work as follows,
In the first pass we use \Cref{lem:KW_clustering} with parameter $p_1$ and $r$ to obtain a spanner $H_1$ and partial partition $\mathcal{P}_1$. We construct a super graph $\mathcal{G}_1$ of $G$ by contracting all internal edges in $\mathcal{P}_1$, and deleting all vertices out of $\cup \mathcal{P}_1$. As $\mathcal{P}_1$ is known after a single pass, the construction of $\mathcal{G}_1$ takes a single pass.

Generally, after $i$ iterations, which took us $i$ passes, we will have spanners $H_1,\dots,H_i$, partition $\mathcal{P}_i$ of $V$ and a super graph $\mathcal{G}_i$ which was constructed by contracting the clusters in $\mathcal{P}_i$, and deleting vertices out of $\cup \mathcal{P}_i$. 
We invoke \Cref{lem:KW_clustering} with parameters $p_i$ and $r$ to obtain a spanner $H_{i+1}$, and partition $\mathcal{P}_{i+1}$.
In \Cref{rem:KWclusteringSuperGraph}, we explain how to use \Cref{lem:KW_clustering} on a super graph $\mathcal{G}_i$ rather than on $G$. Farther, instead of obtaining spanner $\mathcal{H}_{i+1}$ of $\mathcal{G}_i$, we can obtain a spanner $H_{i+1}$ of $G$ such that for every edge $\tilde{e}=(C,C')\in \mathcal{H}_{i+1}$, $H_{i+1}$ contains a representative edge $e\in E(C,C')$.
Then, we create a super graph $\mathcal{G}_{i+1}$ out of  $\mathcal{G}_{i}$ by contracting the clusters in $\mathcal{P}_{i+1}$, and deleting clusters out of $\cup \mathcal{P}_{i+1}$.
Finally, after $g$ passes, we will have a partition 
$\mathcal{P}_{g}$ of $V$. In the $g+1$'th pass, for every pair of clusters $C,C'\in \mathcal{P}_{g}$, we try to sample a single edge from $E(C,C')$ using \Cref{lem:RecoverSingleEdge}. All the sampled edges will be added to a spanner $H_{g+1}$. Note that $|H_{g+1}|\le{|\mathcal{P}_{g}|\choose 2}$. The final spanner $H=\cup_{i=1}^{g+1} H_i$ will be constructed as a union of all the $g+1$ spanners we constructed.

Next, we turn to analyzing the algorithm. First for the number of passes, note that the construction of $\mathcal{P}_i$ is done at the $i$'th pass, while we finish constructing $H_i$ only in the $i+1$'th pass. In particular, in the $i+1$'th pass we will simultaneously construct $H_i$ and  $\mathcal{P}_{i+1}$. This is possible as $\mathcal{P}_i$ (and therefore $\mathcal{G}_i$) is already known by the end of the $i$'th pass.
An exception is $H_{g+1}$ which is computed in a single $g+1$'th pass (where we also simultaneously construct $H_{g}$).

It follows from \Cref{lem:KW_clustering}, that for every $j$, $|\mathcal{P}_{j}|$ is distributed according to $B(n,p^r_1\cdot p^r_2\cdots p^r_j)$,  
thus 
\[
\mathbb{E}\left[\left|\mathcal{P}_{i}\right|\right]=n\cdot\Pi_{q=1}^{j}p_{q-1}=n^{1-r\sum_{q=1}^{j}d_{q}}=n^{1+\frac{1}{k}-d_{j+1}}~.
\]
In particular, using Chernoff inequality (see e.g., thm. 7.2.9.  \href{https://sarielhp.org/misc/blog/15/09/03/chernoff.pdf}{here})
$$\Pr\left[\left|\mathcal{P}_{j}\right|\ge2\cdot n^{1+\frac{1}{k}-d_{j+1}}\right]\le\exp\left(-\frac{1}{4}n^{1+\frac{1}{k}-d_{j+1}}\right).$$
Thus w.h.p. for every $j$, $\left|\mathcal{P}_{j}\right|=O(n^{1+\frac{1}{k}-d_{j+1}})$. The rest of the analysis is conditioned on this bound holding for every $j$.
According to \Cref{lem:KW_clustering}, for every $j$ it holds that 
\[
|H_{j}|\le\tilde{O}(\frac{\left|\mathcal{P}_{j-1}\right|}{p_{j}})=\tilde{O}(n^{1+\frac{1}{k}-d_{j}}\cdot n^{d_{j}})=\tilde{O}(n^{1+\frac{1}{k}})~.
\]
Furthermore, 
\[
|H_{g+1}|\le\left|\mathcal{P}_{g}\right|^{2}=O(n^{2(1+\frac{1}{k}-d_{g+1})})=O(n^{2(1+\frac{1}{k}-\frac{(r+1)^{g}}{k})})\le O(n^{1+\frac{1}{k}})
\]
where the last inequality follows as $1+\frac{1}{k}-\frac{(r+1)^{g}}{k}\le1+\frac{1}{k}-\frac{k+1}{2k}=\frac{1}{2}(1+\frac{1}{k})$.
We conclude the we return a spanner of size $|H|=\tilde{O}(n^{1+\frac{1}{k}})$, and used $\tilde{O}(n^{1+\frac{1}{k}})$ space in every pass.

Finally, we analyze stretch. Denote by $D_i$ the maximal diameter of a cluster in $\mathcal{P}_i$ w.r.t to $G$. Here $\mathcal{P}_0=V$ and thus $D_0=0$. By \Cref{lem:KW_clustering}, the diameter of each cluster in  $\mathcal{P}_i$ w.r.t.  $\mathcal{G}_{i-1}$ is bounded by $\alpha=2^{r+1}-2$. 
As every path inside an $\mathcal{G}_{i}$ cluster will use at most $\alpha$ edges from $H_i$, and will go through at most $\alpha+1$ different clusters in $\mathcal{G}_{i-1}$, it follows that $D_{i+1}\le\alpha+\left(\alpha+1\right)D_{i}$. As $D_0=0$ (singleton clusters), solving this recursion yields
\[
D_{i}=(\alpha+1)^{i}-1=(2^{r+1}-1)^{i}-1~.
\]
Consider a pair of neighboring vertices $u,v$. If there is some cluster at some level $i\le g$ containing both $u,v$, then $d_H(u,v)\le (2^{r+1}-1)^{g}-1$. Else, let $i\in[1,g]$ be the minimal index $i$ such that $\{u,v\}\nsubseteq \cup\mathcal{P}_{i}$ (denote $\mathcal{P}_{g+1}=\emptyset$). 
If $i\le g$, then there is a path in $H_i$ of length $2^{r}-1$ between the clusters in $\mathcal{P}_{i-1}$ containing $u,v$. It thus holds that
\begin{align*}
d_{H}(u,v) & \le(2^{r}-1)+2^{r}\cdot D_{i-1}=2^{r}\cdot\left(1+(2^{r+1}-1)^{i-1}-1\right)-1\\
& =2^{r}\cdot(2^{r+1}-1)^{i-1}-1\le2^{r}\cdot(2^{r+1}-1)^{g-1}-1
\end{align*}
Otherwise, if $i=g+1$, then there is an edge in $H_{g+1}$ between two $\mathcal{P}_{g}$ clusters containing $u$ and $v$. It follows that 
\begin{align}
d_{H}(u,v) & \le1+2\cdot D_{g}=1+2\cdot\left((2^{r+1}-1)^{g}-1\right)\nonumber \\
& =2\cdot(2^{\left\lceil (\frac{k+1}{2})^{\nicefrac{1}{g}}\right\rceil }-1)^{g}-1< 2^{g\cdot k^{\nicefrac{1}{g}}}\cdot2^{g+1}~,\label{eq:KWtransitionStretch}
\end{align}
which is the maximum among the three bounds. The theorem follows.
\end{proof}

Using the same algorithm, while replacing \Cref{lem:KW_clustering} with \Cref{lem:BS_clusters}, we obtain the following,
\thmBSInPasses*
\begin{proof}
	We proceed in the same manner as in \Cref{thm:KWInPasses}, where the only difference is that we use \Cref{lem:BS_clusters} instead of \Cref{lem:KW_clustering}.
	In particular, we use the same parameters $r=\left\lceil (\frac{k+1}{2})^{\nicefrac{1}{g}}\right\rceil -1$, and $d_j,p_j$ as before, for $g$ iterations.
	See \Cref{rem:BSclusteringSuperGraph} for why we can use \Cref{lem:BS_clusters} over a super graph $\mathcal{G}_i$ in this case.
	Similarly, in the last pass we will add an edge between every pair of $\mathcal{P}_g$ clusters.
	It follows from the analysis of \Cref{thm:KWInPasses} that we are using $\tilde{O}(n^{1+\frac1k})$ space and return a spanner with  $\tilde{O}(n^{1+\frac1k})$ edges.
	To analyze the number of passes used, note that in each of the $g$ iterations we need only $r$ passes, where the $r+1$'th pass can be done simultaneously to the first pass in the next iteration. We will also execute the last special pass at the end (together with the $r+1$'th pass of the $g$'th iteration). Thus in total $g\cdot r+1=g\cdot\left(\left\lceil (\frac{k+1}{2})^{\nicefrac{1}{g}}\right\rceil -1\right)+1<g\cdot k^{\nicefrac{1}{g}}+1$ passes.

	To bound the stretch, we will first analyze the maximal diameter $D_i$ of the cluster constituting $\mathcal{G}_i$. It holds that $D_0=0$, while by \Cref{lem:BS_clusters} each cluster in $\mathcal{G}_{i}$ has diameter $2r$ w.r.t. $\mathcal{G}_{i-1}$. Thus $D_{i}\le 2r+(2r+1)D_{i-1} $. Solving this recursion we obtain
	$$D_{i}\le (2r+1)^i-1~.$$
	Consider a pair of neighboring vertices $u,v$. If there is some cluster at some level $i\le g$ containing both $u,v$, then $d_H(u,v)\le (2r+1)^{g}-1$. Else, let $i\in[1,g+1]$ be the minimal index $i$ such that $\{u,v\}\nsubseteq \cup\mathcal{P}_{i}$ (denote $\mathcal{P}_{g+1}=\emptyset$). 
	If $i\le g$, then there is a path in $H_i$ of length $2r-1$ between the  $\mathcal{P}_{i-1}$ clusters containing $u,v$. It thus holds that
	\[
	d_{H}(u,v)\le2r-1+2r\cdot D_{i-1}=2r\cdot(2r+1)^{i-1}-1<(2r+1)^{g}
	\]
	Otherwise, if $i=g+1$, then there is an edge in $H_{g+1}$ between two $\mathcal{P}_{g}$ clusters containing $u$ and $v$. It follows that 
	\begin{align}
	d_{H}(u,v) & \le1+2\cdot D_{g}=1+2\cdot\left((2r+1)^{g}-1\right)\nonumber \\
	& =2\cdot\left(2\cdot\left\lceil (\frac{k+1}{2})^{\nicefrac{1}{g}}\right\rceil -1\right)^{g}-1\approx2^{g}\cdot\left(k+1\right)\label{eq:BStransitionStretch}
	\end{align}
	which is the maximum among the three bounds. The theorem follows.
\end{proof}

\subsection{Corollaries}\label{sec:corollaries}
In this subsection we emphasize some cases of special interest that follow from \Cref{thm:KWInPasses} and \Cref{thm:BSInPasses}. In all the corollaries and the discussion above we discuss a dynamic stream algorithms over an $n$ vertex graphs, that use $\tilde{O}(n^{1+\frac1k})$ space and w.h.p. return a spanner with $\tilde{O}(n^{1+\frac1k})$ edges. 
The performance of the different algorithms is illustrated in \Cref{tab:SpecialCases} for some specific parameter regimes. 

First, surprisingly we obtain a direct improvement over \cite{KW14} and \cite{BS07}. Specifically, in \Cref{cor:BetterKW} we obtain a quadratic improvement in the stretch, while still using only $2$ passes. 
Then, in \Cref{cor:halfBS}, for the case where $k$ is an odd integer, we achieve the exact same parameters as \cite{BS07}, while using only half the number of passes. 
Next, we treat the case where we are allowed $\log k$ passes. Interestingly, for this case, \Cref{thm:KWInPasses} and \Cref{thm:BSInPasses} coincide, and obtain stretch $\approx k^{\log3}$, a polynomial improvement over the $k^{\log5}-1$ stretch in $\log k$ passes by Ahn, Guha, and McGregor \cite{AGM12Spanners}.
Another interesting case is that of $3$ passes. In \Cref{cor:3passes} we show that using a single additional pass compared to \cite{KW14} (and \Cref{cor:BetterKW}), we obtain an exponential improvement in the stretch.
Finally, when one wishes to get close to optimal stretch, in \Cref{cor:sqrtKpasses} we show that compared to \cite{BS07}, we can reduce the number of passes quadratically, while paying only additional factor of $2$ in the stretch.

\begin{corollary}\label{cor:BetterKW}
	There is a $2$ pass algorithm that obtains stretch $2^{\left\lceil \frac{k+1}{2}\right\rceil +1}-3$.
\end{corollary}
\begin{proof}
	Fix $g=1$, then by \Cref{thm:KWInPasses} in two passes we obtain stretch  $2\cdot(2^{\left\lceil \frac{k+1}{2}\right\rceil }-1)-1=2^{\left\lceil \frac{k+1}{2}\right\rceil +1}-3$.
\end{proof}

\begin{corollary}\label{cor:halfBS}
	For an odd integer $k\ge 3$, there is a $\frac{k+1}{2}$-pass algorithm that obtains stretch $2k-1$.
\end{corollary}
\begin{proof}
	Fix $g=1$, then by \Cref{thm:BSInPasses} there is a $\left\lceil \frac{k+1}{2}\right\rceil =\frac{k+1}{2}$ pass algorithm that obtain stretch 
	 $2\cdot\left(2\cdot(\frac{k+1}{2})-1\right)-1=2k-1$.	
\end{proof}

\begin{corollary}\label{cor:logKpasses}
	There is an $\lceil\log(k+1)\rceil$ pass algorithm that obtains stretch $2\cdot3^{\lceil\log\frac{k+1}{2}\rceil}-1\le2\cdot k^{\log3}-1$.
\end{corollary}
\begin{proof}
	Set $g=\lceil\log\frac{k+1}{2}\rceil$, then using we are using $\lceil\log\frac{k+1}{2}\rceil+1=\lceil\log(k+1)\rceil$ passes while having stretch  $2\cdot(2^{\left\lceil (\frac{k+1}{2})^{\nicefrac{1}{g}}\right\rceil }-1)^{g}-1=2\cdot3^{\lceil\log\frac{k+1}{2}\rceil}-1\le2\cdot k^{\log3}-1$.
	
	Interestingly, using the same $g$ in \Cref{thm:BSInPasses}, we will also obtain the exact same result! Specifically, a spanner in $g\cdot\left(\left\lceil (\frac{k+1}{2})^{\nicefrac{1}{g}}\right\rceil -1\right)+1=\lceil\log\frac{k+1}{2}\rceil+1=\lceil\log(k+1)\rceil$
	 passes, with stretch $2\cdot\left(2\cdot\left\lceil (\frac{k+1}{2})^{\nicefrac{1}{g}}\right\rceil -1\right)^{g}-1=2\cdot3^{\lceil\log\frac{k+1}{2}\rceil}-1$.
	
	Following the proof of \Cref{thm:KWInPasses}, set $g=\lceil\log\frac{k+1}{2}\rceil\le\log k$. Then $\left\lceil (\frac{k+1}{2})^{\nicefrac{1}{g}}\right\rceil=2$.
	By \Cref{thm:KWInPasses} we obtain stretch $2\cdot(2^{\left\lceil (\frac{k+1}{2})^{\nicefrac{1}{g}}\right\rceil }-1)^{g}-1=2\cdot3^{\lceil\log\frac{k+1}{2}\rceil}-1\le2\cdot k^{\log3}-1$, while the number of passes is $g+1=\lceil\log(k+1)\rceil$. 
	
	Interestingly, using the same $g$ in \Cref{thm:BSInPasses}, we will also obtain the exact same result! Specifically, a spanner in $g\cdot r+1=g+1=\lceil\log(k+1)\rceil$ passes, while the stretch will be $2\cdot\left(2\cdot\left\lceil (\frac{k+1}{2})^{\nicefrac{1}{g}}\right\rceil -1\right)^{g}-1=2\cdot3^{g}-1$. 
\end{proof}

Note that for $k=3$, both \Cref{cor:BetterKW}, \Cref{cor:halfBS} and \Cref{cor:logKpasses} obtain the best stretch possible: $5$, while using only two passes.

\begin{corollary}\label{cor:3passes}
	There is a $3$ pass algorithm that obtains stretch $2\cdot(2^{\left\lceil \sqrt{\nicefrac{k+1}{2}}\right\rceil }-1)^{2}-1<2^{\sqrt{2k+2}+3}$.
\end{corollary}
\begin{proof}
	Fix $g=2$, then by \Cref{thm:KWInPasses} we have a $3$ pass algorithm with stretch $2\cdot(2^{\lceil \sqrt{\nicefrac{k+1}{2}}\rceil }-1)^{2}-1<2^{\sqrt{2k+2}+3}$.
\end{proof}

\begin{corollary}\label{cor:sqrtKpasses}
	\sloppy There is a $\sqrt{2(k+1)}+1$ pass algorithm that obtains stretch ${8\left\lceil \sqrt{\nicefrac{k+1}{2}}\right\rceil (\left\lceil \sqrt{\nicefrac{k+1}{2}}\right\rceil -1)+1}\approx4k$.
\end{corollary}
\begin{proof}
	Fix $g=2$, then by \Cref{thm:BSInPasses} there is a $2\cdot\lceil \sqrt{\nicefrac{k+1}{2}}\rceil -1$ pass algorithm that obtains stretch
	$2\cdot(2\cdot\lceil\sqrt{\nicefrac{k+1}{2}}\rceil-1)^{2}-1=8\lceil\sqrt{\nicefrac{k+1}{2}}\rceil\cdot(\lceil\sqrt{\nicefrac{k+1}{2}}\rceil-1)+1$.
\end{proof}

\begin{table}[h]
	\begin{tabular}{|c|l|l|l|l|}
		\hline
		& Ref      & Space                          & \#Passes & Stretch    \\ \hline\hline
		\multirow{8}{*}{$k=7$}  & \cite{BS07}     & $\tilde{O}(n^{1+\frac17})$      & 7        & 13         \\ \cline{2-5} 
		& \Cref{cor:halfBS} & $\tilde{O}(n^{1+\frac17})$      & 4        & 13         \\\cline{2-5} 
		& \Cref{cor:sqrtKpasses}    & $\tilde{O}(n^{1+\frac17})$      & 3        & 17         \\ \cline{2-5} 		& \cite{AGM12Spanners} & $\tilde{O}(n^{1+\frac17})$      & 3        & 90         \\ \cline{2-5} 
		& \Cref{cor:logKpasses}    & $\tilde{O}(n^{1+\frac17})$      & 3        & 17         \\ \cline{2-5} 		
		& \Cref{cor:3passes}    & $\tilde{O}(n^{1+\frac17})$      & 3        & 17         \\ \cline{2-5} 
		& \cite{KW14}     & $\tilde{O}(n^{1+\frac17})$      & 2        & 127        \\ \cline{2-5} 
		& \Cref{cor:BetterKW}     & $\tilde{O}(n^{1+\frac17})$      & 2        & 29        \\ \hline\hline
		
		\multirow{8}{*}{$k=31$} & \cite{BS07}     & $\tilde{O}(n^{1+\frac{1}{31}})$ & 31       & 61         \\ \cline{2-5} 
		& \Cref{cor:halfBS} & $\tilde{O}(n^{1+\frac{1}{31}})$ & 16       & 61         \\ \cline{2-5} 
		& \Cref{cor:sqrtKpasses}    & $\tilde{O}(n^{1+\frac{1}{31}})$ & 9        & 97         \\ \cline{2-5} 
		& \cite{AGM12Spanners}    & $\tilde{O}(n^{1+\frac{1}{31}})$ & 5        & 2901         \\ \cline{2-5} 
		& \Cref{cor:logKpasses}    & $\tilde{O}(n^{1+\frac{1}{31}})$ & 5        & 161         \\ \cline{2-5} 				
		& \Cref{cor:3passes}    & $\tilde{O}(n^{1+\frac{1}{31}})$ & 3        & 449       \\ \cline{2-5} 
		& \cite{KW14}     & $\tilde{O}(n^{1+\frac{1}{31}})$ & 2        & $2^{31}-1$ \\ \cline{2-5} 
		& \Cref{cor:BetterKW}     & $\tilde{O}(n^{1+\frac{1}{31}})$ & 2        & $2^{17}-3$ \\ \hline\hline
		
		\multirow{8}{*}{$k=71$} & \cite{BS07}     & $\tilde{O}(n^{1+\frac{1}{71}})$ & 71       & $141\approx2^{7.1}$         \\ \cline{2-5} 
		& \Cref{cor:halfBS} & $\tilde{O}(n^{1+\frac{1}{71}})$ & 36       & $141\approx2^{7.1}$         \\ \cline{2-5} 
		& \Cref{cor:sqrtKpasses}    & $\tilde{O}(n^{1+\frac{1}{71}})$ & 13        & $241\approx2^{7.9}$         \\ \cline{2-5} 
		& \cite{AGM12Spanners}    & $\tilde{O}(n^{1+\frac{1}{71}})$ & 7        & $19882\approx2^{14.3}$         \\ \cline{2-5} 
		& \Cref{cor:logKpasses}    & $\tilde{O}(n^{1+\frac{1}{71}})$ & 7        & $1457<2^{10.5}$        \\ \cline{2-5} 				
		& \Cref{cor:3passes}    & $\tilde{O}(n^{1+\frac{1}{71}})$ & 3        & $7937<2^{13}$       \\ \cline{2-5} 
		& \cite{KW14}     & $\tilde{O}(n^{1+\frac{1}{71}})$ & 2        & $2^{71}-1$ \\ \cline{2-5} 
		& \Cref{cor:BetterKW}     & $\tilde{O}(n^{1+\frac{1}{71}})$ & 2        & $2^{37}-3$ \\ \hline		
	\end{tabular}
	\caption{An illustration of various trade-offs between stretch to the number of passes, for $k=7,31,71$ achieved by different algorithms while using the same space. The parameter $7,31,71$ were chosen to be representatives so that $\sqrt{\frac{k+1}{2}}$ will be an integer.}\label{tab:SpecialCases}
\end{table}

%\newpage
{\small
	\bibliographystyle{alphaurlinit}
	\bibliography{bibArnold}
}
\newpage

\appendix
{\Huge Appendix}
\section{Conjectured hard input distribution}\label{sec:lb-conj}

\newcommand{\D}{\mathcal{D}}
Let $\pi:[n]\to [n]$ be a uniformly random permutation. Let $d>1$ be an integer parameter. Define the distribution $\D'$ on graphs $G=(V, E)$, $V=[n]$ as follows. For every pair $(i, j)\in [n]$ such that $\| \pi(i)-\pi(j)\|_\circ\leq d$ include an edge $(i, j)$ in $E$ with probability $1/2$, where $\|i-j\|_\circ$ is the circular distance on a cycle of length $n$. Define the distribution $\D$ on graphs $G=(V, E)$ as follows. First sample $G'=(V, E')\sim \D'$, and pick two edges $(a, b), (c, d)\sim \mathrm{Unif}(E')$ independently without replacement, and let 
$$
E=(E'\cup \{(a, c), (b, d)\})\setminus \{(a, b), (c, d)\}.
$$ 
Let $G=(V, E)$ be a sample from $\D$. Note that with constant probability over the choice of $G\sim \D$ one has that the distance between $G$ from $a$ to $c$ in  $E\setminus \{(a, c), (b, d)\}$ is $\Omega(n/d)$ and the distance between $c$ and $d$ in $E\setminus \{(a, c), (b, d)\}$ is $\Omega(n/d)$ (see \Cref{fig:GGp} for an illustration). Thus, every $k$-spanner with $k\ll n/d$ must contain both of these edges.  We conjecture that recovering these edges from a linear sketch of the input graph $G$ sampled from $\D$ requires $n^{1+\Omega(1)}$ space when $d=n^{1/3+\Omega(1)}$. Note that the diameter of the graph is (up to polylogarithmic factors) equals $n/d$, and hence this would in  particular imply that obtaining an $n^{2/3-\Omega(1)}$ spanner using a linear sketch requires $n^{1+\Omega(1)}$ bits of space, and therefore imply \Cref{conj:1}.

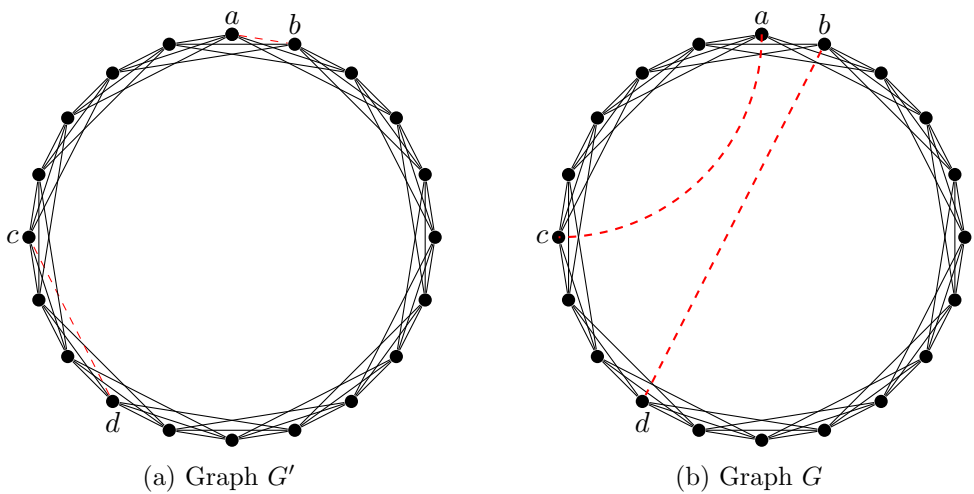
\begin{figure}[H]
	\begin{subfigure}[t]{0.4\textwidth}
		\centering
		\tikzstyle{vertex}=[circle, fill=black!70, minimum size=3,inner sep=1pt]
		\tikzstyle{svertex}=[circle, fill=black!100, minimum size=5,inner sep=1pt]
		\tikzstyle{gvertex}=[circle, fill=green!80, minimum size=7,inner sep=1pt]
		
		\tikzstyle{evertex}=[circle,draw=none, minimum size=25pt,inner sep=1pt]
		\tikzstyle{edge} = [draw,-, color=red!100, very  thick]
		\tikzstyle{bedge} = [draw,-, color=green2!100, very  thick]
		\begin{tikzpicture}[scale=2.7, auto,swap]
		
		\node[svertex] (v0) at ({sin(0)},{cos(0)}){};
		\node[svertex] (v1) at ({sin(18)},{cos(18)}){};
		\node[svertex] (v2) at ({sin(2*18)},{cos(2*18)}){};
		\node[svertex] (v3) at ({sin(3*18)},{cos(3*18)}){};
		\node[svertex] (v4) at ({sin(4*18)},{cos(4*18)}){};
		\node[svertex] (v5) at ({sin(5*18)},{cos(5*18)}){};
		\node[svertex] (v6) at ({sin(6*18)},{cos(6*18)}){};
		\node[svertex] (v7) at ({sin(7*18)},{cos(7*18)}){};
		\node[svertex] (v8) at ({sin(8*18)},{cos(8*18)}){};
		\node[svertex] (v9) at ({sin(9*18)},{cos(9*18)}){};
		\node[svertex] (v10) at ({sin(10*18)},{cos(10*18)}){};
		\node[svertex] (v11) at ({sin(11*18)},{cos(11*18)}){};
		\node[svertex] (v12) at ({sin(12*18)},{cos(12*18)}){};
		\node[svertex] (v13) at ({sin(13*18)},{cos(13*18)}){};
		\node[svertex] (v14) at ({sin(14*18)},{cos(14*18)}){};
		\node[svertex] (v15) at ({sin(15*18)},{cos(15*18)}){};
		\node[svertex] (v16) at ({sin(16*18)},{cos(16*18)}){};
		\node[svertex] (v17) at ({sin(17*18)},{cos(17*18)}){};
		\node[svertex] (v18) at ({sin(18*18)},{cos(18*18)}){};
		\node[svertex] (v19) at ({sin(19*18)},{cos(19*18)}){};
		\node[svertex] (v20) at ({sin(20*18)},{cos(20*18)}){};
		
		\draw[dashed,red] (v0)--(v1);
		\draw (v1)--(v2);
		\draw (v2)--(v3);
		\draw (v3)--(v4);
		\draw (v4)--(v5);
		\draw (v5)--(v6);
		\draw (v6)--(v7);
		\draw (v7)--(v8);
		\draw (v8)--(v9);
		\draw (v9)--(v10);
		\draw (v10)--(v11);
		\draw (v11)--(v12);
		\draw (v12)--(v13);
		\draw (v13)--(v14);
		\draw (v14)--(v15);
		\draw (v15)--(v16);
		\draw (v16)--(v17);
		\draw (v17)--(v18);
		\draw (v18)--(v19);
		\draw (v19)--(v20);

		\draw (v1)--(v3);
		\draw (v2)--(v4);
		\draw (v3)--(v5);
		\draw (v4)--(v6);
		\draw (v5)--(v7);
		\draw (v6)--(v8);
		\draw (v7)--(v9);
		\draw (v8)--(v10);
		\draw (v9)--(v11);
		\draw (v10)--(v12);
		\draw (v11)--(v13);
		\draw (v12)--(v14);
		\draw (v13)--(v15);
		\draw (v14)--(v16);
		\draw (v15)--(v17);
		\draw (v16)--(v18);
		\draw (v17)--(v19);
		\draw (v18)--(v20);
		\draw (v19)--(v1);
		\draw (v20)--(v2);

		\draw (v1)--(v4);
		\draw (v2)--(v5);
		\draw (v3)--(v6);
		\draw (v4)--(v7);
		\draw (v5)--(v8);
		\draw (v6)--(v9);
		\draw (v7)--(v10);
		\draw (v8)--(v11);
		\draw (v9)--(v12);
		\draw (v10)--(v13);
		\draw (v11)--(v14);
		\draw[dashed,red] (v12)--(v15);
		\draw (v13)--(v16);
		\draw  (v14)--(v17);
		\draw (v15)--(v18);
		\draw (v16)--(v19);
		\draw (v17)--(v20);
		\draw (v18)--(v1);
		\draw (v19)--(v2);
		\draw (v20)--(v3);
		
		\draw (v0) [above] node {{$a$}};
		\draw (v1) [above] node {{$b$}};
		\draw (v15) [left] node {{$c$}};
		\draw (v12) [below] node {{$d$}};

		\end{tikzpicture}
		\subcaption{Graph $G'$}
	\end{subfigure}
	\hspace{0cm}
	\begin{subfigure}[t]{0.4\textwidth}
		\centering
		\tikzstyle{vertex}=[circle, fill=black!70, minimum size=3,inner sep=1pt]
		\tikzstyle{svertex}=[circle, fill=black!100, minimum size=5,inner sep=1pt]
		\tikzstyle{gvertex}=[circle, fill=green!80, minimum size=7,inner sep=1pt]
		
		\tikzstyle{evertex}=[circle,draw=none, minimum size=25pt,inner sep=1pt]
		\tikzstyle{edge} = [draw,-, color=red!100, very  thick]
		\tikzstyle{bedge} = [draw,-, color=green2!100, very  thick]
		\begin{tikzpicture}[scale=2.7, auto,swap]
		
		\node[svertex] (v0) at ({sin(0)},{cos(0)}){};
		\node[svertex] (v1) at ({sin(18)},{cos(18)}){};
		\node[svertex] (v2) at ({sin(2*18)},{cos(2*18)}){};
		\node[svertex] (v3) at ({sin(3*18)},{cos(3*18)}){};
		\node[svertex] (v4) at ({sin(4*18)},{cos(4*18)}){};
		\node[svertex] (v5) at ({sin(5*18)},{cos(5*18)}){};
		\node[svertex] (v6) at ({sin(6*18)},{cos(6*18)}){};
		\node[svertex] (v7) at ({sin(7*18)},{cos(7*18)}){};
		\node[svertex] (v8) at ({sin(8*18)},{cos(8*18)}){};
		\node[svertex] (v9) at ({sin(9*18)},{cos(9*18)}){};
		\node[svertex] (v10) at ({sin(10*18)},{cos(10*18)}){};
		\node[svertex] (v11) at ({sin(11*18)},{cos(11*18)}){};
		\node[svertex] (v12) at ({sin(12*18)},{cos(12*18)}){};
		\node[svertex] (v13) at ({sin(13*18)},{cos(13*18)}){};
		\node[svertex] (v14) at ({sin(14*18)},{cos(14*18)}){};
		\node[svertex] (v15) at ({sin(15*18)},{cos(15*18)}){};
		\node[svertex] (v16) at ({sin(16*18)},{cos(16*18)}){};
		\node[svertex] (v17) at ({sin(17*18)},{cos(17*18)}){};
		\node[svertex] (v18) at ({sin(18*18)},{cos(18*18)}){};
		\node[svertex] (v19) at ({sin(19*18)},{cos(19*18)}){};
		\node[svertex] (v20) at ({sin(20*18)},{cos(20*18)}){};
		
		\draw (v1)--(v2);
		\draw (v2)--(v3);
		\draw (v3)--(v4);
		\draw (v4)--(v5);
		\draw (v5)--(v6);
		\draw (v6)--(v7);
		\draw (v7)--(v8);
		\draw (v8)--(v9);
		\draw (v9)--(v10);
		\draw (v10)--(v11);
		\draw (v11)--(v12);
		\draw (v12)--(v13);
		\draw (v13)--(v14);
		\draw (v14)--(v15);
		\draw (v15)--(v16);
		\draw (v16)--(v17);
		\draw (v17)--(v18);
		\draw (v18)--(v19);
		\draw (v19)--(v20);

		\draw (v1)--(v3);
		\draw (v2)--(v4);
		\draw (v3)--(v5);
		\draw (v4)--(v6);
		\draw (v5)--(v7);
		\draw (v6)--(v8);
		\draw (v7)--(v9);
		\draw (v8)--(v10);
		\draw (v9)--(v11);
		\draw (v10)--(v12);
		\draw (v11)--(v13);
		\draw (v12)--(v14);
		\draw (v13)--(v15);
		\draw (v14)--(v16);
		\draw (v15)--(v17);
		\draw (v16)--(v18);
		\draw (v17)--(v19);
		\draw (v18)--(v20);
		\draw (v19)--(v1);
		\draw (v20)--(v2);

		\draw (v1)--(v4);
		\draw (v2)--(v5);
		\draw (v3)--(v6);
		\draw (v4)--(v7);
		\draw (v5)--(v8);
		\draw (v6)--(v9);
		\draw (v7)--(v10);
		\draw (v8)--(v11);
		\draw (v9)--(v12);
		\draw (v10)--(v13);
		\draw (v11)--(v14);
		\draw (v13)--(v16);
		\draw (v14)--(v17);
		\draw (v15)--(v18);
		\draw (v16)--(v19);
		\draw (v17)--(v20);
		\draw (v18)--(v1);
		\draw (v19)--(v2);
		\draw (v20)--(v3);

		\draw (v0) [above] node {{$a$}};
		\draw (v1) [above] node {{$b$}};
		\draw (v15) [left] node {{$c$}};
		\draw (v12) [below] node {{$d$}};

		\draw[dashed,red,thick] (v0)  arc (0:-90:1);
		
		\draw[dashed,red,thick] (v1)  -- (v12);

		\end{tikzpicture}
		\subcaption{Graph $G$}
	\end{subfigure}
	\caption{Illustration of the conjectured hard input distribution} \label{fig:GGp}
\end{figure}

\section{Omitted Proofs}\label{sec:omitted}
\subsection{Omitted proofs from \Cref{sec:prelims}}\label{sec:omitted-2}
\begin{proof}[{\bf Proof of \Cref{lem:SubsetSampling}}]
		First, we start by the following well-known fact about $\ell_0$-sampling sketching algorithms. 
		\begin{fact}[See e.g., \cite{DBLP:conf/pods/JowhariST11,DBLP:conf/focs/KapralovNPWWY17}]\label{fact:ellzero}
			For any vector $\vec{a}\in \R^{n}$ that 
			\begin{itemize}
				\item receives coordinate updates in a dynamic stream,
				\item and each entry is bounded by $O(\poly(n))$,
			\end{itemize} one can design an $\ell_0$-sampler procedure, which succeeds with high probability, by storing a vector $\vec{b}\in \R^{\polylog(n)}$, where 
			\begin{itemize}
				\item (Bounded entries) each entry of $\vec{b}$ is bounded by $O(\poly(n))$,
				\item (Linearity) there exists a matrix $\Pi$ (called \emph{sketching matrix}) such that $\vec{b} = \Pi \cdot \vec{a}$.
			\end{itemize}
		\end{fact} 
		Now, we use \Cref{fact:ellzero} to prepare a data structure for $\ell_0$ sampling of $A_i$ for each $i\in[r]$.\footnote{Note that we do not need to sample uniformly over the non-zero entries, and just recovering a non-zero element is enough for our purpose, however, we still use a $\ell_0$-sampling procedure.} Thus, we are going to have $r$ vectors, $\vec{b}_1,\ldots, \vec{b}_r$, where for each $i\in [r]$ the entries of $\vec{b}_i$ are bounded by $O(\poly(n))$. However, our space is limited to $s\cdot \polylog n$, so we cannot store these vectors. Define $\vec{b}$ as concatenation of $\vec{b}_1,\ldots,\vec{b}_r$. At this point, the reader should note that by assumption $|\mathcal{I}|\le s$, which implies that at most $s\cdot  \polylog (n)$ of entries of $\vec{b}$ are non-zero at the end of the stream.\footnote{However, at some time during the stream, you may have more than $s\cdot\polylog(n)$ non-zeroes, so it is not possible to store $\vec{b}$ explicitly.} Now, we need to prepare a $\wt{O}(s)$-sparse recovery primitive for $\vec{b}$. We apply the following well-known fact about sparse recovery sketching algorithms. 
		\begin{fact}[Sparse recovery]\label{fact:s-sparse}
			For any vector $\vec{b}\in \R^{n}$ that 
			\begin{itemize}
				\item receives coordinate updates in a dynamic stream,
				\item and each entry is bounded by $O(\poly(n))$,
			\end{itemize} one can design a $s$-sparse recovery sketching procedure, by storing a vector $\vec{w}\in \R^{s\cdot\polylog(n)}$, where 
			\begin{itemize}
				\item (Bounded entries) each entry of $\vec{w}$ is bounded by $O(\poly(n))$,
				\item (Linearity) there exists a matrix $\Pi$ (called \emph{sketching matrix}), where $\vec{w} = \Pi \cdot \vec{b}$,
			\end{itemize}
		such that if $\mathrm{Support}(\vec{b})\le s$, then it can recover all non-zero entries of $\vec{b}$.
		\end{fact}
		Using \Cref{fact:s-sparse}, we can store a sketch of $\vec{b}$, in $\wt{O}(s)$ bits of space and recover the non-zero entries of $\vec{b}$ at the end of the stream, which in turn recovers a non-zero element from each $A_i$, in $\vec{v}$. In other words, one can see this procedure as the following linear operation
		\begin{align*}
		\vec{w} = \Pi_1 \Pi_2 \vec{v}
		\end{align*}
		where matrix $\Pi_2\in \R^{r\cdot\polylog(n)\times n}$ is in charge of $\ell_0$ sampling for each $A_i$ and concatenation of vectors, and $\Pi_1\in \R^{(s\cdot \polylog(n))\times(r\cdot \polylog(n))}$ is responsible for the sparse recovery procedure. 

\end{proof}

\begin{proof}[{\bf Proof of \Cref{lem:RecoverSingleEdge}}]
	Let vector $\vec{b}$ be an indicator vector for edges in $A\times B$, i.e., each entry corresponds to a pair of vertices in $A\times B$ and is $1$ if the edge is in the graph, and is $0$ otherwise. Now, by applying \Cref{fact:ellzero} to this vector, one can recover an edge using space $O(\polylog(n))$. Also, using \Cref{fact:s-sparse} with $s=m$, we can recover all $m$ edges using space $m\cdot \polylog(n)$. 
\end{proof}

\subsection{Omitted proofs from \Cref{sec:stream}}\label{sec:omitted-3}
\begin{proof}[{\bf Proof of \Cref{thm:SpectralSparsAreSpannersEdges}}]
	Consider an edge $(u,v)\in E$. Similarly to the proof of \Cref{thm:SpectralSparsAreSpanners}, fix $s:=d_{\hat{H}}(u,v)$, and $A_{i}:=\{z\in V\mid d_{\hat{H}}(v,z)=i\}$  for $i\in [0,s-1]$ be
	all the vertices at distance $i$ from $v$ in $\hat{H}$. Set $A_{s}:=\{z\in V\mid d_{\hat{H}}(v,z)\ge s\}$
	to be all the vertices at distance at least $s$ from $v$.
	In addition set $W_{i}^{H}=w_H(A_i\times A_{i+1})$ and $W_{i}^{G}=w_G(A_i\times A_{i+1})$. 
	Also recall that $W_{-1}^{H}=W_{-1}^{G}=W_{s}^{H}=W_{s}^{G}=0$.
	
	We will follow steps similar to those in \Cref{claim:sumWH}. 
	Set 
	\begin{align}\label{eq:alpha2} \alpha=\Theta(\eps\log n)\text{ such that } \alpha\ge10\log_{\frac{1}{6\epsilon}}\frac{2s}{\alpha}
	,\end{align} and $I=\left\{ i\in[0,s-1]\mid W_{i}^{G}\le\frac{\alpha m}{s}\right\} $.
	It holds that $|I|\ge\left(1-\frac{1}{\alpha}\right)s+1$, as otherwise
	there are more than $\frac{s}{\alpha}$ indices $i$ for which $W_{i}^{G}>\frac{\alpha m}{s}$,
	implying $\sum_{i}W_{i}^{G}>m$, a contradiction, since $\{W_{i}^{G}\}_i$ represent the number of elements in disjoint sets of edges. Set $$\wt{I}=\left\{ i\mid\text{such that }\forall j,\,|i-j|\le\frac{\alpha}{10}\text{ it holds that }j\in I\right\}.$$
	Then there are less than $\frac{s}{\alpha}\cdot\frac{2\alpha}{10}<\frac{s}{2}$
	indices out of $\wt{I}$, implying 
	\begin{align}\label{eq:tildeI2}
	\left|\wt{I}\right|\ge\frac{s}{2}.
	\end{align}
	
	For any index $i\in\wt{I}$ and any index $j\in \left[i-\frac{\alpha}{10},i+\frac{\alpha}{10}-1\right]$, by \Cref{claim:CutHG}, $$W_{i-1}^{H}+W_{i}^H+W_{i+1}^{H}\ge\frac{1}{\epsilon}\left(W_{i}^{H}-W_{i}^{G}\right)\ge\frac{1}{\epsilon}\left(W_{i}^{H}-\frac{\alpha m}{s}\right).$$
	Assume for contradiction that $W_{i}^{H}>2\cdot\frac{\alpha m}{s}$.
	Then, 
	\[
	W_{i-1}^{H}+W_{i}^H+W_{i+1}^{H}>\frac{1}{\epsilon}\left(\frac{\alpha m}{s}-\frac{\alpha m}{s}\right)=\frac{1}{\epsilon}\cdot\frac{\alpha m}{s}~.
	\]
	Let $i_{1}\in\left\{ i-1,i,i+1\right\} $ such that $W_{i_{1}}^{H}\ge\frac{1}{3\epsilon}\cdot \frac{\alpha m}{s}\ge \frac{1}{6\epsilon}\cdot \frac{\alpha m}{s}$.
	Using the same argument, 
	\[
	W_{i_{1}-1}^{H}+W_{i_{1}}^{H}+W_{i_{1}+1}^{H}\ge\frac{1}{\epsilon}\left(W_{i_{1}}^{H}-\frac{\alpha m}{s}\right)>\frac{1}{2\epsilon}\cdot\frac{1}{6\epsilon}\cdot\frac{\alpha m}{s}~.
	\]
	Choose $i_{2}\in\left\{ i_{1}-1,i_{1},i_{1}+1\right\} $ such that
	$W_{i_{2}}^{H}>\frac{1}{(6\epsilon)^{2}}\cdot\frac{\alpha m}{s}$.
	As $i\in\wt{I}$, we can continue this process for $\frac{\alpha}{10}$
	steps, where in the $j$ step we have $W_{i_{j}}^{H}>\frac{1}{(6\epsilon)^{j}}\cdot\frac{\alpha m}{s}$.
	In particular $$W_{i_{\frac{\alpha}{10}}}^{H}>\left(6\epsilon\right)^{-\frac{\alpha}{10}}\frac{\alpha m}{s}\ge2m,$$ a contradiction, as $H$
	is an $(1\pm\epsilon)$-spectral sparsifier of the unweighted graph $G$,
	where the maximal size of a cut is $m$. We conclude
	that for every $i\in\wt{I}$ it holds that $W_{i}^{H}\le2\cdot\frac{\alpha m}{s}$.
	It follows that 
	\begin{align}\label{eq:sumWiH}
	\sum_{i=0}^{s-1}\frac{1}{W_{i}^{H}}&\ge\left|\wt{I}\right|\cdot\frac{s}{2\alpha m}\nonumber\\
	&\ge\frac{s^{2}}{4\alpha m}&&\text{By \Cref{eq:tildeI2}}\nonumber\\
	&=\wt{\Omega}\left(\frac{s^{2}}{m}\right)&&\text{By setting of $\alpha$ in \Cref{eq:alpha2}}
	\end{align}
	
	Construct an auxiliary graph $H'$ from $H$, by  contracting all the vertices inside each set $A_{i}$, and keeping multiple edges. Note that by this operation, the effective resistance between $u$ and $v$ can only decrease.
	The graph $H'$ is a path graph consisting of $s$ vertices, where the conductance between the $i$'th vertex to the $i+1$'th is $W_{i}^H$. We conclude
	\begin{align}
	(1+\epsilon)R_{u,v}^{G} & \ge R_{u,v}^{H}&&\text{By \Cref{fact:EffSparse}}\nonumber\\
	&\ge R_{u,v}^{H'} &&\text{As explained above}\nonumber\\
	&=\sum_{i=0}^{s-1}\frac{1}{W_{i}^{H}} &&\text{Since $H'$ is a path graph}\nonumber\\
	&=\wt{\Omega}\left(\frac{s^{2}}{m}\right)&&\text{By \Cref{eq:sumWiH}}\label{eq:RGboundUsingM}
	\end{align}
	As $u,v$ are neighbors in the unweighted graph $G$, it necessarily holds that $R_{u,v}^{G}\le 1$, implying that $s=\wt{O}\left(\sqrt{m}\right)$.
\end{proof}

\section{Baswana-Sen \cite{BS07} spanner}\label{sec:altBS}
Originally Baswana and Sen constructed $2k-1$ spanners with $\tilde{O}(n^{1+\frac1k})$ edges in the sequential setting. Assuming Erd\H{o}s girth conjecture, this construction is optimal up to second order terms.
Ahn \etal \cite{AGM12Spanners} adapted the spanner of \cite{BS07} to the dynamic-stream framework using $\tilde{O}(n^{1+\frac1k})$ space and $k$ passes. 
We begin this section with the sequential algorithm of \cite{BS07}. Then, we will provide it's streaming implementation by \cite{AGM12Spanners}, with a proof sketch. Afterwards, we will state the clustering  \Cref{lem:BS_clusters} that follows from the analysis of this algorithm, with some discussion.
Interestingly, for odd integers $k$, in \Cref{cor:halfBS}, using the same clustering technique we obtain a spanner with the same performance as \cite{BS07}, while using only half the number of passes.

\begin{theorem}[\cite{BS07}+\cite{AGM12Spanners}]\label{thm:BS}
	Given an integer $k\ge 1$, there is a $k$-pass algorithm, that given the edges of an $n$-vertex graph in a dynamic stream fashion, using  $\tilde{O}(n^{1+\frac1k})$ space, w.h.p. constructs a $2k-1$-spanner with $\tilde{O}(n^{1+\frac1k})$ edges.
\end{theorem}
\begin{proof}
We will start with a sequential description of the algorithm, which is also illustrated in \Cref{alg:BS07}. Afterwards, we will explain how to implement this algorithm is the streaming model, and we will finish with an analysis of its performance.

\paragraph{Sequential spanner construction.}
Initially $H=\emptyset$. 
The algorithm runs in $k$ steps. We have $k+1$ sets $V=N_0\supseteq N_1 \supseteq \dots\supseteq N_{k-1}\supseteq N_{k}=\emptyset$.
For $i<k$, each vertex $v\in N_{i-1}$, joins $N_{i}$ with probability $p=n^{-\frac1k}$.
In each stage we will have set of clusters, rooted in $N_i$. Initially we have $n$ singleton clusters. For $v\in N_i$, it will be the root of clusters (or trees) $T_{v,0}\subseteq T_{v,1} \subseteq \dots\subseteq  T_{v,i}$. 
In stage $i$, for each vertex $v\in  T_{u,i-1}$ that belong to an $i-1$ cluster do as follows: 
If $u\in N_i$, that is $v$ also belongs to an $i$ cluster, do nothing. 
Else ($u\in N_{i-1}\setminus N_i$), look for an edge from $v$ towards $\cup_{z\in N_{i}}T_{z,i-1}$, that is towards an $i-1$ cluster that becomes an $i$ cluster. If there is such an edge $e_v$, towards $T_{z,i}$, $v$ joins $T_{z,i}$ and $e_v$ is added to $H$. Otherwise, go over all the clusters $\{T_{z,i-1}\}_{z\in N_{i}}$, and add a single crossing edge from $v$ to each one of them (if exist).
Note that if $v$ did not belong to any $i-1$ cluster we do nothing.

\begin{algorithm}[t]
	\caption{\texttt{Sequential spanner construction: ala \cite{BS07}}}	\label{alg:BS07}
	\DontPrintSemicolon
	\SetKwInOut{Input}{input}\SetKwInOut{Output}{output}
	\Input{$n$ vertex graph $G=(V,E)$, parameter $k$}
	\Output{$2k-1$ spanner $H$ with $\tilde{O}(n^{1+\frac1k})$ edges}
	\BlankLine
	Set $N_0=V$ and $N_k=\emptyset$. For every $v\in N_0$ set $T_{v,0}\leftarrow \{v\}$\;
	\For{$i=1$ to $k-1$}{
		$N_i\leftarrow \emptyset$\;
		\ForEach{$v\in N_{i-1}$}{
			i.i.d. with probability $n^{-\nicefrac{1}{k}}$ add $v$ to $N_i$
		}
	}
	\For{$i=1$ to $k$}{
		\ForEach{$v\in N_{i}$}{Set $T_{v,i}\leftarrow T_{v,i-1}$}
		\ForEach{$v\in (\cup_{u\in N_{i-1}}T_{u,i-1})\setminus(\cup_{u\in N_{i}}T_{u,i-1})$}{
			Sample an edge $e_v=(v,y)\in \{v\}\times\cup_{u\in N_{i}}T_{u,i-1}$\;
			\If{$e_v\neq\emptyset$}{
				Add $e_v$ to $H$\;
				Let $u\in N_{i}$ s.t. $y\in T_{u,i-1}$, add $v$ to $T_{u,i}$\;
			}
			\Else{
				\ForEach{$u\in N_{i-1}$}{
					Sample an edge $e_v\in 	\{v\}\times T_{u,i-1}$\;
					\textbf{if} $e_v\neq\emptyset$ \textbf{then} add $e_v$ to $H$
				}
			}
		}
	}
	\Return $H$\;
\end{algorithm}

\paragraph{Streaming implementation.}
Each step of the algorithm is implemented in a single streaming pass. In the $i$'th pass, for every vertex $v\in\cup_{z\in N_{i-1}\setminus N_{i}}T_{z,i-1}$, using \Cref{lem:RecoverSingleEdge} we will sample an edge $e_v=(v,y)\in \{v\}\times \cup_{z\in N_{i}}T_{z,i-1}$. This will determine whether $v$ joins an $i$-cluster.
In addition, for each such vertex  $v\in\cup_{z\in N_{i-1}\setminus N_{i}}T_{z,i-1}$, we will sample $\tilde{O}(n^{\frac1k})$ edges from the star graph $G_{v,i-1}$ defined as follows: the set of nodes will be $\{v\}\cup N_{i-1}$ where there is an edge from $v$ to $z\in N_{i-1}$ in $G_{v,i-1}$ iff in $G$ there is an edge from $v$ to a vertex in $T_{z,i-1}$. Note that we can interpret the edge stream for $G$ as an edge stream for $G_{v,i-1}$ (by ignoring all non-relevant edges). Thus we can use \Cref{lem:RecoverSingleEdge}.
In case $v$ did not joined $i$ cluster, next in the $i+1$ pass, for every sampled edge $(v,z)$ from $G_{v,i-1}$, we will sample an edge $e_{v,z}$ from $\{v\}\times T_{z,i-1}$ using \Cref{lem:RecoverSingleEdge} and add it to $H$. \\
For the last stage, $N_k=\emptyset$, for each vertex $v\in\cup_{z\in N_{k-1}}T_{z,k-1}$, instead of looking for an neighbor in an $k$ cluster, we will simply sample a single edge from $v$ the each cluster in $\{T_{z,k-1}\}_{z\in N_{k-1}}$ (using \Cref{lem:RecoverSingleEdge}), and add it to $H$.

\paragraph{Analysis.}
We start by bounding the space, and number of edges. 
First note than perhaps for the last round, according to \Cref{lem:RecoverSingleEdge} we are using at most $\tilde{O}(n^{\frac1k})$ space per vertex per round, and thus a total of  $\tilde{O}(n^{1+\frac1k})$.
Considering the last round, set 
$\mu=\mathbb{E}[|N_{k-1}|]=n^{1-\frac{k-1}{k}}=n^{\frac{1}{k}}$, by Chernoff inequality (see e.g. thm. 7.2.9.  \href{https://sarielhp.org/misc/blog/15/09/03/chernoff.pdf}{here}), $\Pr[\left||N_{k-1}| - \mu\right| \geq \mu+O(\log n)] =\poly(\frac 1n)$. Hence w.h.p.  $|N_{k-1}|=\tilde{O}(n^{\frac{1}{k}})$. If this event indeed occurred, in the last $k$'th round we will be using additional 
$\tilde{O}\left(n\cdot|N_{k}|\right)=\tilde{O}\left(n^{1+\frac{1}{k}}\right)$ space to sample edges towards the last level clusters.

For a vertex $z$, if $G_{z,i-1}$ contains $\Omega( n^{\frac1k}\log n)$ vertices (in other words there are at least $\Omega(n^{\frac1k}\log n)$ $i-1$ cluster containing a neighbor of $z$), then the probability that $z$ will fail to join an $i$ cluster is bounded by $(1-n^{-\frac{1}{k}})^{\Omega(n^{\frac{1}{k}}\log n)}=n^{-\Omega(1)}$. Thus we will assume that for every level $i\le k-1$ and vertex $z$ such that  $G_{z,i-1}$ contains $\Omega( n^{\frac1k}\log n)$, $z$ will have an edge towards an $i$ cluster. By \Cref{lem:RecoverSingleEdge} w.h.p. we will sample such an edge. In the other case, by \Cref{lem:RecoverSingleEdge} again w.h.p. we will manage to sample all the edges in $G_{z,i-1}$.
It follows that using $\tilde{O}(n^{1+\frac1k})$ space we manage to implement \Cref{alg:BS07} fatefully. In particular $H$ contains $\tilde{O}(n^{1+\frac1k})$ edges.

It remains to analyze the stretch. By induction, it easily follows that for each $i$ cluster $T_{v,i}$ has radius at most $i$ w.r.t. $v$ in $H$. Consider an edge $(x,y)$ in $G$. Suppose that the highest level cluster $x$ (resp. $y$) belongs to is $T_{v,i}$ (resp. $T_{u,j}$) where w.l.o.g $i\le j$. If $T_{v,i}=T_{u,j}$ then by the radius bound $d_H(x,y)\le d_H(x,v)+d_H(v,y)\le 2i\le 2k-2$.
Else, we added an edge from $x$ to a vertex $y'$ belonging to an $i$ cluster $T_{u',i}$ containing $y$. Thus 
\begin{equation}
d_{H}(x,y)\le d_{H}(x,y')+d_{H}(y',u')+d_{H}(u',y)\le1+2\cdot i\le1+2\cdot(k-1)=2k-1~.\label{eq:BSstretch}
\end{equation}
\end{proof}

\paragraph{\cite{BS07} clustering} We can stop the running of \Cref{alg:BS07} after $i+1$ iterations for some $i<k$. In fact, we can do this even if $k\ge1$ is not integer. we conclude:
\BSclustering*
\begin{proof}[Proof sketch.]
	We run the algorithm of \Cref{thm:BS} for $i+1$ rounds, where each vertex $v\in N_{j-1}$ joins $N_j$ with probability $p$. 
	Here $\mathcal{P}=\{T_{v,i}\}_{v\in N_i}$. Thus indeed $|\mathcal{P}|$ distributed according to $B(|V|,p^i)$. It follows from the analysis of \Cref{thm:BS}, that each cluster in $T_{v,i}\in\mathcal{P}$ has radius $i$, and thus diameter $2i$. Finally, according to \cref{eq:BSstretch}, if $x\notin \cup\mathcal P$, then for every $y$, $d_H(x,y)\le 1+2(i-1)=2i-1$.
\end{proof}
\begin{remark}\label{rem:BSclusteringSuperGraph}[Super graph clustering]
	During the algorithm of \Cref{thm:BSInPasses}, we actually use \Cref{lem:BS_clusters} for a super graph $\mathcal{G}$ of $G$ rather than for the actual graph. Specifically, there is a partial partition of $G$ into clusters $\mathcal{C}$, and there is an edge between clusters $C,C'\in\mathcal{C}$ in $\mathcal{G}$ if and only if $E(C,C')\ne \emptyset$.\\
	We argue that \Cref{lem:BS_clusters} can be used in this regime as well.
	First, given such a representation of a super graph using partial partition $\mathcal{C}$, we can treat a stream of edges for $G$ as a stream of edges for $\mathcal{G}$. Specifically, when seeing an insertion/deletion of an edge $e=(u,v)$: if either $u,v$ belong to the same cluster, or one of them doesn't belong to a cluster at all- simply ignore $e$. Otherwise, simulate insertion/deletion the edge $\tilde{e}=(C_1,C_2)$, where $C_1,C_2\in\mathcal{C}$ are the clusters containing $u,v$.\\
	Second, even though initially we suppose to receive a spanner $\mathcal{H}$ of $\mathcal{G}$, we can actually instead obtain for every edge $\tilde{e}=(C,C')\in \mathcal{H}$, a representative edge $e\in E(C,C')$. To see this, note that in \Cref{alg:BS07} there are two types of edges added to $\mathcal{H}$. Consider $C\in T_{\tilde{C},j-1}$. Then in the $j$'th pass, $C$ ``will try'' to join a $j$ cluster, specifically we sample a single edge from $C$ towards $\cup_{C'\in N_{j}}T_{C',j-1}$, and also $\tilde{O}(\frac{1}{p})$ edges in the auxiliary graph $G_{C,j-1}$. If we manage to sample an edge towards $\cup_{C'\in N_{j}}T_{C',j-1}$, than we can sample a representative in $G$ for this edge in the next $j+1$'th pass. 	
	Else, in the $j+1$'th pass the algorithm will sample a representative  for each edge in $G_{C,j-1}$. Observe, that as the algorithm samples a representative edges between clusters in $\mathcal{G}$, say from $C$ to $T_{C',j-1}$ we actually can instead sample an edge between the actual clusters in $G$, $C,\bigcup T_{C',j-1}\subset V$.	
\end{remark}

\section{Kapralov-Woodruff \cite{KW14} Spanner}\label{sec:KW}
Kapralov and Woodruff constructed a spanner in $2$ passes of a dynamic stream, with stretch $2^k-1$ using   $\tilde{O}(n^{1+\frac1k})$ space. Their basic approach is similar to \cite{BS07}, where the difference is that all the clustering steps are done in a single pass, using the linear nature of $\ell_0$ samplers.
As a result, the diameter of an $i$-level cluster is blown up from $2i$ to $2^{i+1}-2$.
We begin this section by providing the details of \cite{KW14} algorithm. Afterwards, we will state the clustering \Cref{lem:KW_clustering} that follows from the analysis of this algorithm, with some discussion.
Surprisingly, in \Cref{cor:BetterKW}, using the same clustering technique, in $2$ passes only using the same space, we obtain a quadratic improvement in the stretch compared to \cite{KW14}.

\begin{algorithm}[t]
	\caption{\texttt{\cite{KW14} sequential spanner construction}}\label{alg:KW14}
	\DontPrintSemicolon
	\SetKwInOut{Input}{input}\SetKwInOut{Output}{output}
	\Input{$n$ vertex graph $G=(V,E)$, parameter $k$}
	\Output{$2^k-1$ spanner $H$ with $\tilde{O}(n^{1+\frac1k})$ edges}
	\BlankLine
	Set $H=\emptyset$, $N_0=V$ and $N_k=\emptyset$\;
	\For{$i=1$ to $k-1$}{
		$N_i=\emptyset$\;
		\For{$v\in V$}{
			With probability $n^{-\frac ik}$, add $v$ to $N_i$, and set $T_{v,i}=\{v\}$\;
		}
	}
	\For{$i=1$ to $k$}{
		\ForEach{$v\in N_{i-1}$}{
			Sample an edge $e_v=\{x,u\}\in T_{v,i-1}\times N_i$\;
			\If{$e_v\neq\emptyset$}{
				Add $e_v$ to $H$\;
				$T_{u,i}\leftarrow T_{u,i}\cup T_{v,i-1}$}
			\Else{
				\ForEach{vertex $z\in\mathcal{N}(T_{v,i-1})$}{
					Sample an edge $e_z\in 	T_{v,i-1}\times \{z\}$,
					add $e_v$ to $H$.
				}
			}
	}}
	\Return $H$\;
\end{algorithm}
\begin{theorem}[\cite{KW14}]\label{thm:KW}
	For every integer $k\ge 1$, there is a 2 pass dynamic stream algorithm that given an unweighted, undirected $n$-vertex graph $G=(V,E)$, uses $\tilde{O}(n^{1+\frac1k})$ space, and computes w.h.p. a spanner $H$ with   $\tilde{O}(n^{1+\frac1k})$ edges and stretch $2^k-1$.	
\end{theorem}
\begin{proof}
	We begin by providing a sequential version of \cite{KW14} algorithm, which is also illustrated in \Cref{alg:KW14}. Then we will show how to implement it in 2 passes of a dynamic stream and sketch the analysis.	Given a cluster $C\subseteq V$, we denote by $\mathcal{N}(C)=\left\{u\in V\setminus C\mid E\cap(u\times C)\ne\emptyset\right\}$ the set of vertices out of $C$ with a neighbor in $C$.
	
\paragraph{Sequential spanner construction.}
There are two steps: clustering, and adding edges between clusters.
Initially $H=\emptyset$. Sample sets $N_0,N_1,\dots,N_{k-1},N_k$ as follows: $N_0=V$ and $N_k=\emptyset$. Each vertex $v\in V$ joins $N_i$ i.i.d. with probability $n^{-\frac ik}$. Note that the sets are not necessarily nested.
For $v\in N_i$ set $T_{v,i}=\{v\}$. 
We will have $k-1$ clustering steps. Initially each vertex $v$ belongs to a $0$-level singleton cluster $T_{v,0}$. In general, for level $i$ we will have a collection of $i-1$-level clusters $\{T_{v,i-1}\}_{v\in N_{i-1}}$, and will construct $i$-level clusters. $N_{i}$ will be the centers of this clusters. For each $v\in N_{i-1}$, we will pick a random edge $e_v=(x,u)\in T_{v,i}\times N_{i+1}$ (if exist, if $v\in N_i$ it can also pick itself). Then $e_v$ will be added to $H$, and all the vertices in $T_{v,i-1}$ will join $T_{u,i}$. If no such edge exist, we say that $T_{v,i-1}$ is a terminal cluster. Denote by $\mathcal{I}_{i-1}\subseteq N_{i-1}$ the set of centers of terminal clusters. For each $v\in \mathcal{I}_{i-1}$, add to $H$ a single edge from $T_{v,i-1}$ to every vertex in $\mathcal{N}(T_{v,i})$, the neighbors of $T_{v,i}$. 
Note that every $k-1$ cluster, is a terminal cluster, thus $\mathcal{I}_{k-1}=N_{k-1}$.
See \Cref{alg:KW14}.

\paragraph{Streaming implementation} There will be two passes. In the first we will create the clusters of all the levels. In the second pass we will add edge from the terminal clusters to their neighbors.
The sets $N_1,N_2,\dots,N_{k-1}$ are sampled before the first pass. During the first pass we will maintain an $\ell_0$ sampler from each vertex $v$ to each set $N_i$, i.e. for the sets $v\times N_1, v\times N_2,\dots, v\times N_{k-1}$.
Due to the linear nature of $\ell_0$ samplers, consider an $i-1$-cluster $T_{v,i-1}$, given samplers for $\{u\times N_{i}\}_{u\in T_{v,i-1}}$ we can sample w.h.p. an edge from $T_{v,i-1}\times N_{i}$ (see \Cref{fact:ellzero}). In particular, either $T_{v,i-1}$ will join an $i$ level cluster, or $v$ will join $\mathcal{I}_{i-1}$. 
We can preform all the $k-1$ steps of clustering after the first pass. 

During the second pass, for every index $i\in[0,k-1]$, and every $v\in \mathcal{I}_i$ we will use \Cref{lem:SubsetSampling} with parameter $s=\tilde{O}(n^{\frac{i+1}{k}})$ to sample edges from $T_{v,i}$ to every neighbor in $\mathcal{N}(T_{v,i})$. Specifically, we can think on edges as a vector $\vec{v}\in\R^{|T_{v,i}|\times |V\setminus T_{v,i}|}$ where every pair in $T_{v,i}\times (V\setminus T_{v,i})$ has a representative coordinate. This coordinates are divided to $|V\setminus T_{v,i}|$ sets in the natural way, where the goal is to sample a non empty coordinate (i.e. edge) from each non-empty set of coordinates.
The edges corresponding to the sampled coordinates will be added to the spanner.

\paragraph{Analysis sketch.}
Using Chernoff bound (see e.g. thm. 7.2.9.  \href{https://sarielhp.org/misc/blog/15/09/03/chernoff.pdf}{here}), w.h.p. for every index $i\in[1,k-1]$, $|N_i|=\tilde{O}(n^{1-\frac ik})$. The number $\ell_0$ samplers used during the first pass is $n\cdot k$, thus the overall space used is $\tilde{O}(n)$. 
For a cluster $T_{v,i}$, if $\mathcal{N}(T_{v,i})$ then w.h.p. (again using Chernoff) $v\notin\mathcal{I}_i$ (as each vertex in $\mathcal{N}(T_{v,i})$ joins $N_{i+1}$ independently with probability $|N_i|=O(n^{-\frac ik})$). Hence using \Cref{lem:SubsetSampling} we will indeed succeed in recovering an edge to each neighbor in $\mathcal{N}(T_{v,i})$. We conclude that the streaming algorithm faithfully implemented \Cref{alg:KW14}. The total space (and hence also number of edges) used in the second pass is bounded by 
\[
\sum_{i=0}^{k-1}\left|\mathcal{I}_{i}\right|\cdot\tilde{O}(n^{\frac{i+1}{k}})\le\sum_{i=0}^{k-1}\left|N_{i}\right|\cdot\tilde{O}(n^{\frac{i+1}{k}})=\sum_{i=0}^{k-1}\cdot\tilde{O}(n^{1-\frac{i}{k}})\cdot\tilde{O}(n^{\frac{i+1}{k}})=\tilde{O}(n^{1+\frac{1}{k}})~.
\]

Regarding stretch, we argue by induction that the radius of $T_{v,i}$ in $H$ w.r.t. $v$ is bounded by  $2^i-1$. Indeed it holds for $i=0$ as each $T_{v,0}$ is a singleton.
For the induction step, consider a cluster $T_{v,i}$, and let $z\in T_{v,i}$ be some vertex. If $v\in N_{i-1}$ and $z\in T_{v,i}$ the bound follows from the induction hypothesis. Otherwise, there is some center $u\in N_{i-1}$ such that $z\in T_{u,i-1}$, and $H$ contains an edge from some vertex $x\in T_{u,i-1}$ to $v$. We conclude
$d_H(z,v)\le d_H(z,u)+d_H(u,x)+d_H(x,v)\le 2(2^{i-1}-1)+1=2^{i}-1$.
Next consider an edge $(x,y)$ in $G$. If there is some terminal cluster $T_{v,i}$ containing both $x,y$ then $d_H(x,y)\le d_H(x,v)+d_H(v,y)\le 2\cdot(2^i-1)\le 2^k-2$.
Else, let $i$ be the minimal number such that either $x$ or $y$ belong to a terminal cluster. By minimality there are $v_x,v_y\in N_i$ such that $x\in T_{v_x,i}$ and $y\in T_{v_y,i}$. W.l.o.g. $T_{v_x,i}$ is a terminal cluster. In particular the algorithm adds an edge towards $y$ from some vertex $z\in T_{v_x,i}$. We conclude,
\begin{equation}
d_{H}(x,y)\le d_{H}(x,v_{x})+d_{H}(v_{x},z)+d_{H}(z,y)\le2(2^{i}-1)+1=2^{i+1}-1\le2^{k}-1~.\label{eq:KWstretch}
\end{equation}
\end{proof}

For our construction, we will run \cite{KW14} algorithm for $i$ steps only. The result is the following:

\paragraph{\cite{KW14} clustering} We can stop the running of \Cref{alg:KW14} after $i+1$ iterations for some $i<k$. In fact, we can do this even if $k\ge1$ is not integer. we conclude:
\KWclustering*
\begin{proof}[Proof sketch.]
	We run the first pass in the algorithm of \Cref{alg:KW14} for $i$ rounds, where each vertex $v\in N_{j-1}$ joins $N_j$ with probability $p$. 
	Here $\mathcal{P}=\{T_{v,i}\}_{v\in N_i}$. Thus indeed $|\mathcal{P}|$ distributed according to $B(|V|,p^i)$. It follows from the analysis of \Cref{thm:KW}, that each cluster in $T_{v,i}\in\mathcal{P}$ has radius $2^i-1$, and thus diameter $2^{i+1}-2$. Finally, if $x\notin \cup\mathcal P$, then $x$ belongs to an $i-1$-level terminal cluster. Hence according to \cref{eq:KWstretch}, for every neighbor $y$ of $x$ in $G$, $d_H(x,y)\le 2^{(i-1)+1}-1=2^{i}-1$.
\end{proof}
\begin{remark}\label{rem:KWclusteringSuperGraph}[Super graph clustering]
	Similarly to our usage of \Cref{lem:BS_clusters} discusses in \Cref{rem:BSclusteringSuperGraph}, here as well during the algorithm of \Cref{thm:KWInPasses}, we actually use \Cref{lem:KW_clustering} for a super graph $\mathcal{G}$ of $G$ rather than for the actual graph. Specifically, there will be a partial partition of $G$ into clusters $\mathcal{C}$, and  $\mathcal{G}$ will be defined over  $\mathcal{C}$, where there is an edge between clusters $C,C'\in\mathcal{C}$ in $\mathcal{G}$ if and only if $E(C,C')\ne \emptyset$.
	See \Cref{rem:BSclusteringSuperGraph} for an explanation of why we can treat a stream of edges over $G$, as a stream over $\mathcal{G}$.\\
	Note that even though initially we suppose to receive a spanner $\mathcal{H}$ of $\mathcal{G}$, we can actually instead obtain for every edge $\tilde{e}=(C,C')\in \mathcal{H}$, a representative edge $e\in E(C,C')$. 
	For edges add to the spanner during the first pass, we can simply sample a representations in the second pass. 
	During the second pass for each terminal cluster $T_{C,i}$ we added an edges towards every neighbor in $\mathcal{N}(T_{C,i})$ using \Cref{lem:SubsetSampling}. Specifically we have the sets $\{T_{C,i}\times \{C'\}\}_{C'\in \mathcal{C}\setminus T_{C,i}}$ and sampled a single $\mathcal{G}$ edge from each non-empty set. But this just correspond to edges between actual clusters in $G$. Thus instead we can use \Cref{lem:SubsetSampling} to sample a single $G$ edge from each non-empty set  $\{(\bigcup T_{C,i})\times C'\}_{C'\in \mathcal{C}\setminus T_{C,i}}$.	
\end{remark}
\begin{remark}\label{rem:ImpossibleKWinSCM}
	While \Cref{alg:KW14} can be implemented in the dynamic steaming model in two passes  using $\widetilde{O}(n^{1+\frac1k})$ space, it is impossible to do so in the simultaneous communication model where each player can send only $\widetilde{O}(n^{\frac1k})$ size message in each communication round.
	Specifically, the problem is that there is no equivalent to \Cref{lem:SubsetSampling} in the simultaneous communication model. 
	In more detail, note that for each terminal cluster  $T_{v,i}\in \mathcal{N}(T_{v,i})$ the algorithm might restore $\Omega(n^{\frac{i+1}{k}})$ outgoing edges from $T_{v,i}$. In particular, all this edges might be incident on small number of vertices (even one). In the simultaneous communication model it will be impossible to restore them all.
\end{remark}

\end{document}